\documentclass [journal,onecolumn,11pt]{IEEEtran}
\usepackage{amsfonts,amsmath,amssymb}
\usepackage{indentfirst, setspace}
\usepackage{url,float}
\usepackage{color}
\usepackage[a4paper,ignoreall]{geometry}
\usepackage{longtable}
\usepackage{graphicx}
\usepackage[a4paper,ignoreall]{geometry}
\usepackage{multicol}
\usepackage{stfloats}
\usepackage{enumerate}
\usepackage{cite}
\usepackage{amsthm}
\usepackage{mathrsfs}
\usepackage{multirow}
\usepackage{amssymb}
\usepackage[all]{xy}
\usepackage[overload]{empheq}
\usepackage[square, comma, sort&compress, numbers]{natbib}
\usepackage{float}
\usepackage{booktabs}
\usepackage{diagbox}
\usepackage{multirow}
\usepackage{makecell}
\usepackage{threeparttable}
\usepackage{bookmark}
\usepackage{bm}
\usepackage{empheq}
\usepackage{amsthm,amsmath,amssymb}
\usepackage{array,booktabs,longtable,multirow,makecell,pdflscape}
\newcolumntype{P}[1]{>{\raggedright\arraybackslash}p{#1}}
\newcolumntype{C}[1]{>{\centering\arraybackslash}p{#1}}
\newcolumntype{P}[1]{>{\raggedright\arraybackslash}p{#1}}
\allowdisplaybreaks[4]

\newtheorem{theorem}{Theorem}[section]

\newtheorem{proposition}[theorem]{Proposition}

\newtheorem{lemma}[theorem]{Lemma}
\newtheorem{definition}[theorem]{Definition}
\newtheorem{example}{Example}

\newtheorem{remark}{Remark}

\makeatletter
\newcommand{\figcaption}{\def\@captype{figure}\caption}
\newcommand{\tabcaption}{\def\@captype{table}\caption}
\makeatother

\begin{document}

    \title{Minimum Distances of Binary Goppa Codes and Constructions with Prescribed Alternating Automorphism Groups}
   \author{
    {Tianni He, Kangquan Li, Longjiang Qu}
    \thanks{Tianni He, Kangquan Li, and Longjiang Qu are with the College of Science, National University of Defense Technology, Changsha, 410073, China. E-mail:  hetianni@foxmail.com, likangquan11@nudt.edu.cn, ljqu\_happy@hotmail.com. This work is supported by the National Key Research and Development Program of China under Grant 2024YFA101300, and the National Natural Science Foundation of China (NSFC) under Grants 12525115, 12571579. Hunan Provincial Natural Science Foundation of China under Grant 2026JJ40001.}
    }
    
    \maketitle{}

    \begin{abstract} 
Goppa codes are a well-known class of linear codes with important applications in cryptography. Determining the minimum distance of Goppa codes and constructing Goppa codes with prescribed automorphism groups are both meaningful and challenging problems in coding theory. In this paper, we first study the minimum distance of binary separable Goppa codes. For the two classes \(g(X)=f(X^t)\) and \(g(X)=A(X)h(\phi(X))\), we give criteria for attaining the designed distance and derive several infinite families whose minimum distances are determined. We then construct binary Goppa codes and their related codes with \(A_4\) or \(A_5\) automorphism groups. These constructions also naturally yield binary quasi-cyclic Goppa codes and their related codes. Moreover, by applying the minimum-distance criteria developed above, we determine the parameters of one class of the constructed \(A_4\)-invariant Goppa codes.   
    \end{abstract}


    \section{Introduction}
    Goppa codes play an important role in coding theory and code-based cryptography \cite{Goppa1970,Goppa1971,2-21,HuffmanPless2003}. In code-based cryptosystems, selecting an appropriate code involves balancing public-key size, security, error-correction capability and so on. The McEliece public-key cryptosystem, one of the earliest code-based cryptosystems, was originally built from binary irreducible Goppa codes \cite{McEliece1978}. Its modern variant, Classic McEliece, remains one of the most important candidates in post-quantum cryptography \cite{NIST2022}. 
    
    Although Goppa codes have shown strong resistance to many known structural attacks, their use in code-based cryptography is often constrained by large public-key sizes. This motivates the study of structured subclasses of Goppa codes and related alternant codes, especially those with nontrivial automorphism groups, since such symmetries may lead to more compact descriptions and more efficient implementations \cite{Gaborit2005,BergerCayrelGaboritOtmani2009,MisoczkiBarreto2009,BarretoLindnerMisoczki2011,Persichetti2012}. Besides this cryptographic motivation, constructing Goppa codes with special algebraic structures and studying their parameters are also meaningful and challenging problems in algebraic coding theory.
    
    The structure and properties of Goppa codes have been studied from several perspectives, including connections to cyclic codes \cite{Berger1999,Berger2000a,BommierBlanchet2000}, enumeration problems \cite{15,9,10}, decoding algorithms for twisted Goppa codes \cite{22,Sun}, and dimension analysis \cite{VanDerVlugt1990,VanDerVlugt1991,Veron2001,Veron2005,QuanYue2024}. Although there have been many studies on Goppa codes, the determination of the minimum distance remains a formidable challenge.
    In 1992,  C. J. Moreno and O. Moreno \cite{MorenoMoreno1992} used exponential sums to construct a subclass of Goppa codes whose minimum distance is equal to \(2t+1\). In 1995,  S. V. Bezzateev and N. A. Shekhunova \cite{BezzateevShekhunova1995} proved that, for separable Goppa polynomials \(G(X)=X^t+A\), where \(A\) is a \(t\)-th power in \(\mathbb F_{2^m}\) and \(t\mid(2^m-1)\), the minimum distance attains the designed distance \(2t+1\). In 1998, P. V\'eron \cite{Veron1998} studied Goppa codes defined by trace operators and obtained parity information on their minimum distances. In 2008,  S. V. Bezzateev and N. A. Shekhunova \cite{BezzateevShekhunova2008} proved that certain binary separable Goppa codes form a chain and that the minimum distances of all codes in the chain can be determined. In 2024, Y. Wu et al. \cite{WuLiHuYue2024} extended the result of \cite{BezzateevShekhunova1995} to the family \(G(X)=X^{3t}+1\), under the assumption \(t\mid(2^m-1)\), and proved that the minimum distance is \(6t+1\) when certain explicit conditions on \(t\) are satisfied.
    
    Beyond the study of parameters, motivated in part by cryptographic applications, the construction of Goppa codes and related alternant codes with prescribed automorphisms has also attracted attention. T. P. Berger studied cyclic alternant codes induced by automorphisms of generalized Reed--Solomon codes and investigated Goppa and related codes invariant under prescribed permutations \cite{Berger1999,Berger2000a,Berger2000b}.  Building on T. P. Berger's work, X. Li and Q. Yue \cite{LiYue2022Dihedral} constructed binary expurgated and extended Goppa codes with dihedral automorphism groups via \(D_{2n}\)-orbits in \(PGL(2,2^m)\), but gave few examples and did not study their parameters.


The main contributions of this paper are twofold. First, we develop two criteria for determining when binary separable Goppa codes attain their designed distance. The first criterion applies to power-composite Goppa polynomials \(g(X)=f(X^t)\), where multiplicative cosets are used to construct the support, see Theorem \ref{d1}. The second applies to composite Goppa polynomials \(g(X)=A(X)h(\phi(X))\), where complete fibers of the polynomial map \(\phi\) are used, see Theorem \ref{d2}.  As applications, we obtain several infinite families of binary Goppa codes with determined minimum distance, including the families defined by \(g(X)=X^{2t}+X^t+1\), \(g(X)=X^{3t}+X^t+1\), and \(g(X)=X^4+cX\). We summarize the specific results in Table~\ref{tab:main-distance-families}; for all three families, the support is \(L=\{\alpha\in\mathbb F_{2^m}:g(\alpha)\ne0\}\), and in the table \(2^m-1=t\ell\) and \(c\in\mathbb F_{2^m}^*\).

\begin{table}[htbp]
\centering
\begingroup
\normalsize
\setlength{\tabcolsep}{5.5pt}
\renewcommand{\arraystretch}{1.12}
\caption{Three explicit families of Goppa codes $\Gamma(L,g)$ with determined minimum distance}
\label{tab:main-distance-families}
\begin{tabular}{c c c c c}
\toprule
No. & \(g(X)\) & Conditions & \(d\) & Reference \\
\midrule
\(1\) & \(X^{2t}+X^t+1\) & \(5\mid\ell\) or \(7\mid\ell\) & \(4t+1\) & Proposition~\ref{d11}\\
\(2\) & \(X^{3t}+X^t+1\) & \(15\mid\ell\), \(21\mid\ell\), or \(31\mid\ell\) & \(6t+1\) & Proposition~\ref{d12}\\
\(3\) & \(X^4+cX\) & \(m\ge6\) even & \(9\) & Theorem~\ref{thm:x4cx-distance}\\
\bottomrule
\end{tabular}
\endgroup
\end{table}

Second, we study the construction of quasi-cyclic Goppa codes. 
 We construct binary Goppa codes, expurgated Goppa codes and extended Goppa codes with \(A_4\) or \(A_5\) automorphism groups by using the orbit decompositions of \(A_4\)- and \(A_5\)-subgroups of \(PGL(2,2^m)\) on \(\overline{\mathbb F}_{2^m}\), see Theorems~\ref{5}, \ref{2}, and \ref{thm:A5-construction}. These constructions also naturally yield binary quasi-cyclic Goppa codes and their related codes. 
Compared with previous constructions of quasi-cyclic Goppa codes, our
constructions are based on the alternating groups \(A_4\) and \(A_5\) and yield explicit orbit-based families. Moreover, by applying the minimum-distance criterion developed in the first part, we determine the exact parameters
\([\,2^m-4,\;2^m-4m-4,\;9\,]\) of the representative \(A_4\)-invariant
binary Goppa code \(\Gamma(L,g)\), where \(g(X)=X^4+X\) and
\(L=\mathbb F_{2^m}\setminus\mathbb F_4\), see Theorems~\ref{main}.

The rest of this paper is organized as follows. In Section~\ref{1}, we review some basic notation and necessary preliminaries. In Section~\ref{sec3}, we establish two minimum-distance criteria for binary separable Goppa codes and derive several infinite families with determined minimum distance. In Section~\ref{sec4}, we construct binary Goppa codes and related codes with alternating automorphism groups, including the \(A_4\)- and \(A_5\)-constructions, and determine parameters for a representative \(A_4\)-invariant family. In Section~\ref{sec5}, we conclude the paper.

\section{Preliminaries}\label{1}
In this paper, we always assume that \(m\) is a positive integer. We denote by \(\mathbb F_{2^m}\) the finite field with \(2^m\) elements, by \(\mathbb F_{2^m}^*\) its multiplicative group, and by \(\overline{\mathbb F}_{2^m}=\mathbb F_{2^m}\cup\{\infty\}\) a coordinate set for the projective line. We will introduce some basic knowledge in the following subsections.

\subsection{Goppa Codes, Expurgated and Extended Goppa Codes}
First, we recall the definitions of Goppa codes, expurgated Goppa codes, and extended Goppa codes.

\begin{definition}\label{27}\cite{coding, HuffmanPless2003,2-21}
Let \(g(X)=\sum_{i=0}^{r}g_iX^i\in\mathbb F_{2^m}[X]\) be a polynomial of degree \(r\), where \(g_r\ne0\), and let \(L=(\alpha_1,\ldots,\alpha_n)\) be an \(n\)-tuple of distinct elements of \(\mathbb F_{2^m}\) such that \(g(\alpha_i)\ne0\) for \(i=1,2,\ldots,n\). For \(\mathbf c=(c_1,\ldots,c_n)\in\mathbb F_2^n\), let
\[
R_{\mathbf c}(X)=\sum_{i=1}^{n}\frac{c_i}{X-\alpha_i}.
\]
The Goppa code \(\Gamma(L,g)\) is defined as
\[
\Gamma(L,g)=\{\mathbf c\in\mathbb F_2^n:R_{\mathbf c}(X)\equiv0 \pmod{g(X)}\}.
\]
The \(n\)-tuple \(L\) is called the support of the code, and the polynomial \(g(X)\) is called the Goppa polynomial.

The expurgated Goppa code \(\widetilde{\Gamma}(L,g)\) of \(\Gamma(L,g)\) is defined by
\[
\widetilde{\Gamma}(L,g)=
\left\{\mathbf c=(c_1,\ldots,c_n)\in\Gamma(L,g):\sum_{i=1}^{n}c_i=0\right\}.
\]
The extended Goppa code \(\overline{\Gamma}(L,g)\) is defined by
\[
\overline{\Gamma}(L,g)=
\left\{\mathbf c=(c_1,\ldots,c_n,c_{n+1})\in\mathbb F_2^{n+1}:
(c_1,\ldots,c_n)\in\Gamma(L,g),\ \sum_{i=1}^{n+1}c_i=0\right\}.
\]
\end{definition}

The parity-check matrices of these codes are given as follows.

\begin{proposition}\label{47}\cite{coding, HuffmanPless2003,2-21}
With the notation in Definition~\ref{27}, the Goppa code \(\Gamma(L,g)\) has a parity-check matrix
\[
H=
\begin{pmatrix}
g(\alpha_1)^{-1} & g(\alpha_2)^{-1} & \cdots & g(\alpha_n)^{-1}\\
\alpha_1g(\alpha_1)^{-1} & \alpha_2g(\alpha_2)^{-1} & \cdots & \alpha_ng(\alpha_n)^{-1}\\
\vdots & \vdots & \ddots & \vdots\\
\alpha_1^{r-1}g(\alpha_1)^{-1} & \alpha_2^{r-1}g(\alpha_2)^{-1} & \cdots & \alpha_n^{r-1}g(\alpha_n)^{-1}
\end{pmatrix}.
\]
Moreover, the expurgated Goppa code \(\widetilde{\Gamma}(L,g)\) has a parity-check matrix
\[
\widetilde H=
\begin{pmatrix}
g(\alpha_1)^{-1} & g(\alpha_2)^{-1} & \cdots & g(\alpha_n)^{-1}\\
\alpha_1g(\alpha_1)^{-1} & \alpha_2g(\alpha_2)^{-1} & \cdots & \alpha_ng(\alpha_n)^{-1}\\
\vdots & \vdots & \ddots & \vdots\\
\alpha_1^{r}g(\alpha_1)^{-1} & \alpha_2^{r}g(\alpha_2)^{-1} & \cdots & \alpha_n^{r}g(\alpha_n)^{-1}
\end{pmatrix}.
\]
If \(g(\infty)=g_r\), then the extended Goppa code \(\overline{\Gamma}(L,g)\) has a parity-check matrix
\[
\overline H=
\begin{pmatrix}
g(\alpha_1)^{-1} & g(\alpha_2)^{-1} & \cdots & g(\alpha_n)^{-1} & 0\\
\alpha_1g(\alpha_1)^{-1} & \alpha_2g(\alpha_2)^{-1} & \cdots & \alpha_ng(\alpha_n)^{-1} & 0\\
\vdots & \vdots & \ddots & \vdots & \vdots\\
\alpha_1^{r-1}g(\alpha_1)^{-1} & \alpha_2^{r-1}g(\alpha_2)^{-1} & \cdots & \alpha_n^{r-1}g(\alpha_n)^{-1} & 0\\
\alpha_1^{r}g(\alpha_1)^{-1} & \alpha_2^{r}g(\alpha_2)^{-1} & \cdots & \alpha_n^{r}g(\alpha_n)^{-1} & g(\infty)^{-1}
\end{pmatrix}.
\]
\end{proposition}

\subsection{Minimum distance of Goppa codes}
First, we introduce the definition of the minimum distance in linear codes.
\begin{definition}\cite{coding}
Let $\mathbf{x}$ and $\mathbf{y}$ be words of length $n$ over an alphabet $A$. The (Hamming) distance from $\mathbf{x}$ to $\mathbf{y}$, denoted by $d(\mathbf{x}, \mathbf{y})$, is defined to be the number of places at which $\mathbf{x}$ and $\mathbf{y}$ differ. If $\mathbf{x} = (x_1, \ldots, x_n)$ and $\mathbf{y} = (y_1, \ldots ,y_n)$, then 
\[
d(\mathbf{x}, \mathbf{y}) = d(x_1, y_1) + \cdots + d(x_n, y_n), 
\]
where $x_i$ and $y_i$ are regarded as words of length $1$, and 
\[
d(x_i, y_i) = 
\begin{cases} 
1 & \text{if } x_i \neq y_i \\
0 & \text{if } x_i = y_i.
\end{cases}
\]
\end{definition}

\begin{definition}\cite{coding}
  For a code $C$ containing at least two words, the (minimum) distance of $C$, denoted by $d(C)$, is 
\[
d(C) = \min\{d(\mathbf{x}, \mathbf{y}) : \mathbf{x}, \mathbf{y} \in C,\ \mathbf{x} \neq \mathbf{y}\}.
\]  
\end{definition}
Concerning the minimum distance of Goppa codes, we have the following two lemmas.
\begin{lemma} \cite{2-21}\label{26}
Let \(g(X)\in\mathbb F_{2^m}[X]\) be a Goppa polynomial of degree \(r\), and let
\(L=\{\alpha_1,\ldots,\alpha_n\}\subseteq\mathbb F_{2^m}\) be a support. Then the binary Goppa code \(\Gamma(L,g)\) has parameters \([n,k,d]\), where \(k\ge n-mr\) and \(d\ge r+1\).
\end{lemma}

\begin{lemma} \cite{2-21}\label{41}
With the notation of Lemma~\ref{26}, if \(g(X)\) is separable, that is, if it has no multiple roots over the algebraic closure of \(\mathbb F_{2^m}\), then \(d\ge 2r+1\). The number \(2r+1\) is called the designed distance of the binary separable Goppa code \(\Gamma(L,g)\).
\end{lemma}
\subsection{Code Automorphisms and Group Actions}
We first recall the notions of permutation automorphisms and quasi-cyclic codes.
\begin{definition}
Let \(C\) be a linear code of length \(n\), and let \(\psi\) be a permutation of \(\{1,\ldots,n\}\). If \(\mathbf{c}=(c_1,c_2,\ldots,c_n)\in C\), then
\[
\mathbf{c}^{\psi}=(c_{\psi(1)},c_{\psi(2)},\ldots,c_{\psi(n)}).
\]
The permutation \(\psi\) is called a permutation automorphism of \(C\) if \(\mathbf{c}^{\psi}\in C\) for every \(\mathbf{c}\in C\). The group of all permutation automorphisms of \(C\) is called the permutation automorphism group of \(C\).
\end{definition}

\begin{definition}
Let \(C\) be a linear code of length \(n\) and $G$ the permutation
group of \(C\). We say that \(C\) is quasi-cyclic if \(G\) contains a subgroup isomorphic to \(\mathbb Z/\lambda\mathbb Z\), where \(1<\lambda\le n\).
\end{definition}

We also recall the alternating groups used in this paper. Let \(S_n\) be the symmetric group on \(n\) letters. The alternating group \(A_n\) is the subgroup of \(S_n\) consisting of all even permutations. 

Next, we introduce the action of the projective linear group on \(\overline{\mathbb F}_{2^m}\). 
 The general linear group of degree $2$ over ${\mathbb F}_{2^m}$ is
\[
GL(2,2^m)=
\left\{
\begin{pmatrix}
a & b\\
c & d
\end{pmatrix}
:
a,b,c,d\in\mathbb F_{2^m},\ ad-bc\ne 0
\right\}.
\]
The projective general linear group of degree $2$ over ${\mathbb F}_{2^m}$ is defined by
\[
PGL(2,2^m)=GL(2,2^m)/\{kI_2:k\in\mathbb F_{2^m}^*\},
\]
where \(I_2\) is the \(2\times 2\) identity matrix. Thus, two nonsingular matrices define the same element of \(PGL(2,2^m)\) if they differ by a nonzero scalar multiple.
 For
\[
M=
\begin{pmatrix}
a & b\\
c & d
\end{pmatrix}
\in PGL(2,2^m),
\]
it acts on \(\overline{\mathbb F}_{2^m}\) by the fractional linear transformation
\[
M(\zeta)=\frac{a\zeta+b}{c\zeta+d}.
\]
Here the usual conventions are used:
\[
M(\infty)=
\begin{cases}
a/c, & c\ne 0,\\
\infty, & c=0,
\end{cases}
\qquad
M(-d/c)=\infty\quad(c\ne0).
\]

\section{Two Frameworks for Determining the Minimum Distance of Binary Separable Goppa Codes}\label{sec3}
In this section, we investigate the minimum distance of two classes of binary separable Goppa codes. The first class is defined by Goppa polynomials of the form \(g(X)=f(X^t)\), where multiplicative cosets in \(\mathbb F_{2^m}^*\) are used to construct the support. The second class is defined by Goppa polynomials of the form \(g(X)=A(X)h(\phi(X))\), where complete fibers of the polynomial map \(\phi\) are used to construct the support. These constructions make it possible, in certain cases, to explicitly construct codewords attaining the designed distance and hence to determine the exact minimum distance of the corresponding Goppa codes.

\subsection{Power-Composite Goppa Polynomials \(g(X)=f(X^t)\)}

\begin{lemma}\label{3}
Let \(g(X)=f(X^t)\), where \(f(X)\in\mathbb F_{2^m}[X]\), \(f(0)\ne0\), and \(t>0\) is odd. Then \(g(X)\) is separable if and only if \(f(X)\) is separable.
\end{lemma}

\begin{proof}
Since \(t\) is odd, in characteristic \(2\) we have
\[
g'(X)=t X^{t-1}f'(X^t)=X^{t-1}f'(X^t).
\]
If \(g\) is not separable, then there exists \(\alpha\) in the algebraic closure such that \(g(\alpha)=g'(\alpha)=0\). Since \(f(0)\ne0\), we have \(\alpha\ne0\). Let \(y=\alpha^t\). Then \(f(y)=0\) and \(f'(y)=0\), so \(f\) is not separable.

Conversely, if \(f\) is not separable, then there exists \(y\ne0\) in the algebraic closure such that \(f(y)=f'(y)=0\). Choose \(\alpha\) with \(\alpha^t=y\). Then \(g(\alpha)=f(\alpha^t)=f(y)=0\), and
\[
g'(\alpha)=\alpha^{t-1}f'(\alpha^t)=\alpha^{t-1}f'(y)=0.
\]
Thus, \(g\) is not separable. Hence \(g\) is separable if and only if \(f\) is separable.
\end{proof}

We now give a general construction theorem.

\begin{theorem}\label{d1}
Let \(f(X)\in\mathbb F_{2^m}[X]\) be separable with \(f(0)\ne0\), and let \(\deg f=r\). Let \(g(X)=f(X^t)\), where \(t\mid(2^m-1)\), and set
\[
L=\{\alpha\in\mathbb F_{2^m}:g(\alpha)\ne0\},\quad 2^m-1=t\ell.
\]
Suppose that there exists a monic polynomial \(P(X)\in\mathbb F_{2^m}[X]\) of degree \(2r\) such that
\begin{enumerate}
\item \(P(X)\mid X^\ell-1\),
\item \(\gcd(P(X),f(X))=1\),
\item \(f(X)\mid P(X)+XP'(X)\).
\end{enumerate}
Then the minimum distance of the binary Goppa code $\Gamma(L,g)$ equals its designed distance, i.e., 
$d=2rt+1.$
\end{theorem}
\begin{proof}
By Lemma~\ref{3}, the polynomial \(g(X)=f(X^t)\) is separable. Hence, by Lemma~\ref{41}, the binary Goppa code \(\Gamma(L,g)\) satisfies \(d\ge 2\deg g+1=2rt+1\). Thus, it remains to construct a codeword of weight \(2rt+1\).

Since  \(\ell\mid 2^m-1\), \(\ell\) is odd, and therefore \(X^\ell-1\) has no multiple roots over \(\mathbb F_{2^m}\). Since \(P(X)\mid X^\ell-1\) and \(\deg P=2r\), the \(2r\) roots of \(P(X)\) are distinct. Moreover, as \(\ell\mid 2^m-1\), all roots of  \(P(X)\) lie in \(\mathbb F_{2^m}^*\). Denote the roots of \(P(X)\) by \(y_1,\ldots,y_{2r}\).

Since \(P(X)\mid X^\ell-1\), each root \(y_i\) of \(P\) satisfies \(y_i^\ell=1\). As \(2^m-1=t\ell\), each such \(y_i\) is a \(t\)-th power in \(\mathbb F_{2^m}^*\). Hence, for each \(i\), there exists \(x_i\in\mathbb F_{2^m}^*\) such that \(x_i^t=y_i\). We now construct a subset \(S\subseteq L\) of size \(1+2rt\). Let \(w\) be a generator of \(\mathbb F_{2^m}^*\), and put
\[
H=\{u\in\mathbb F_{2^m}^*:u^t=1\}=\langle w^\ell\rangle .
\]
 Let
\[
S=\{0\}\cup\bigcup_{i=1}^{2r}x_iH .
\]
The cosets \(x_iH\) are pairwise disjoint, and therefore \(|S|=1+2rt\). By \(\gcd(P,f)=1\), we have \(f(y_i)\ne0\) for all \(i\). Thus, for every \(\alpha\in x_iH\), we have \(g(\alpha)=f(\alpha^t)=f(y_i)\ne0\), and also \(g(0)=f(0)\ne0\). Hence \(S\subseteq L\).

Next, we prove that the binary vector whose nonzero positions correspond exactly to the elements of \(S\) is a codeword of \(\Gamma(L,g)\). Let
\[
\sigma_S(X)=\prod_{\beta\in S}(X-\beta).
\]
Then
\[
\sigma_S(X)=X\prod_{i=1}^{2r}\prod_{h\in H}(X-x_ih).
\]
Since \(\prod_{h\in H}(X-x_ih)=X^t-x_i^t=X^t+y_i\), we obtain
\[
\sigma_S(X)=X\prod_{i=1}^{2r}(X^t+y_i)=XP(X^t).
\]
Thus
\[
\sigma_S'(X)=P(X^t)+X^tP'(X^t).
\]
By the assumption \(f(X)\mid P(X)+XP'(X)\), substituting \(X^t\) for \(X\) gives
\[
g(X)=f(X^t)\mid P(X^t)+X^tP'(X^t)=\sigma_S'(X).
\]
Writing \(S=\{\beta_1,\ldots,\beta_{2rt+1}\}\), and using \(S\subseteq L\), we have
\(g(\beta_j)\ne0\) for all \(j\). Hence \(\gcd(\sigma_S(X),g(X))=1\). It follows that
\[
\sum_{j=1}^{2rt+1}\frac{1}{X-\beta_j}
=\frac{\sigma_S'(X)}{\sigma_S(X)}
\equiv0\pmod{g(X)}.
\]
Therefore, the binary vector supported exactly on \(S\) is a codeword of \(\Gamma(L,g)\) of weight \(2rt+1\). Hence \(d\le 2rt+1\). Combining this with the lower bound \(d\ge2rt+1\), we obtain \(d=2rt+1\).
\end{proof}



\begin{remark}
The condition \(t\mid(2^m-1)\) forces \(t\) to be odd. In fact, if \(t=2u\), \(u\mid(2^m-1)\), and \(2^m-1=u\ell\), then a similar argument yields \(d=rt+1\) whenever \(\ell\) satisfies the corresponding conditions above. This follows from the fact that \(\Gamma(L,g)=\Gamma(L,g^2)\) for square-free binary Goppa polynomials \cite{2-21}.
\end{remark}

Theorem~\ref{d1} transforms the problem of finding a binary separable Goppa code attaining the designed distance into the problem of finding an auxiliary polynomial \(P(X)\). Compared with a direct search for \(2\deg g+1\) support points in \(\mathbb F_{2^m}\), this condition is often easier to verify. The theorem also provides a systematic way to obtain infinite families of Goppa codes with determined minimum distance. We give two typical infinite families below.

\begin{proposition}\label{d11}
Let \(t\mid 2^m-1\), and write \(2^m-1=t\ell\). Let
\[
g(X)=X^{2t}+X^t+1,\quad
L=\{\alpha\in\mathbb F_{2^m}:g(\alpha)\ne0\}.
\]
If \(5\mid\ell\) or \(7\mid\ell\), then the binary Goppa code \(\Gamma(L,g)\) has minimum distance
$d=4t+1.$   
\end{proposition}

\begin{proof}
Take \(f(X)=X^2+X+1\). Then \(f(0)\ne0\), and since \(f'(X)=1\), the polynomial \(f(X)\) is separable. Moreover, \(g(X)=f(X^t)\) and \(\deg f=2\).

If \(5\mid\ell\), take
\[
P(X)=X^4+X^3+X^2+X+1.
\]
Then \(P(X)\mid X^5-1\), and hence \(P(X)\mid X^\ell-1\). Moreover, \(\gcd(P(X),f(X))=1\), and a direct calculation gives
\[
P(X)+XP'(X)=X^4+X^2+1=f(X)^2.
\]
Thus, \(P(X)\) satisfies the three conditions in Theorem~\ref{d1}.

If \(7\mid\ell\), take
\[
P(X)=X^4+X^2+X+1.
\]
The same direct verification shows that \(P(X)\) satisfies the three conditions in Theorem~\ref{d1}. Therefore, the conclusion follows from Theorem~\ref{d1}.
\end{proof}
\begin{remark}
Under the hypotheses of Proposition~\ref{d11}, assume further that \(m\) is even and \(t\le 2^{m/2-2}\). Then Proposition~\ref{prop:vdv-dimension} gives the exact dimension of \(\Gamma(L,g)\). More precisely, if \(3\mid \ell\), then the code has parameters \([\,2^m-2t,\ 2^m-2t-2mt,\ 4t+1\,]\); if \(3\nmid \ell\), then the code has parameters \([\,2^m,\ 2^m-2mt,\ 4t+1\,]\).
\end{remark}
\begin{example}
Table~\ref{tab:d11-examples} gives examples for Proposition~\ref{d11}. Over \(\mathbb F_{2^m}\), we take \(g(X)=X^{2t}+X^t+1\) and \(L=\{\alpha\in\mathbb F_{2^m}:g(\alpha)\ne0\}\), with \(\ell=(2^m-1)/t\). The parameters of \(\Gamma(L,g)\) were computed by Magma, and the listed minimum distances agree with Proposition~\ref{d11}.

\begin{table}[htbp]
\centering
\begingroup
\setlength{\tabcolsep}{5.5pt}
\renewcommand{\arraystretch}{1.12}
\caption{Parameters of \(\Gamma(L,g)\) for \(g(X)=X^{2t}+X^t+1\)}
\label{tab:d11-examples}
\begin{tabular}{c c c c c}
\toprule
\(m\) & \(t\) & \(\ell\) &  \(g(X)\) & Parameters \\
\midrule
\(4\) & \(1\) & \(15\) & \(X^2+X+1\) & \([14,6,5]\)\\
\(8\) & \(1\) & \(255\) & \(X^2+X+1\) & \([254,238,5]\)\\
\(8\) & \(3\) & \(85\) & \(X^6+X^3+1\) & \([256,208,13]\)\\
\(6\) & \(1\) & \(63\) & \(X^2+X+1\) & \([62,50,5]\)\\
\(6\) & \(3\) & \(21\) & \(X^6+X^3+1\) & \([58,22,13]\)\\
\bottomrule
\end{tabular}
\endgroup
\end{table}
    
\end{example}
\begin{proposition}\label{d12}
Let \(t\mid 2^m-1\), and write \(2^m-1=t\ell\). Let
\[
g(X)=X^{3t}+X^t+1,\quad
L=\{\alpha\in\mathbb F_{2^m}:g(\alpha)\ne0\}.
\]
Suppose that there exists a polynomial of the form
\[
P(X)=X^6+aX^5+bX^3+X^2+cX+1\in\mathbb F_{2^m}[X]
\]
such that \(P(X)\mid X^\ell-1\) and \(\gcd(P(X),X^3+X+1)=1\). Then the binary Goppa code \(\Gamma(L,g)\) has minimum distance
$d=6t+1.$
In particular, if \(15\mid\ell\), or \(21\mid\ell\), or \(31\mid\ell\), then \(d=6t+1\).
\end{proposition}

\begin{proof}
Take \(f(X)=X^3+X+1\). Then \(f(0)\ne0\) and the polynomial \(f(X)\) is separable. Moreover, \(g(X)=f(X^t)\) and \(\deg f=3\). For
\[
P(X)=X^6+aX^5+bX^3+X^2+cX+1,
\]
a direct calculation gives
\[
P(X)+XP'(X)=X^6+X^2+1=f(X)^2.
\]
Thus, as long as \(P(X)\mid X^\ell-1\) and \(\gcd(P(X),f(X))=1\), the polynomial \(P(X)\) satisfies the three conditions in Theorem~\ref{d1}. Hence \(d=2\deg f\cdot t+1=6t+1\).

In particular, take
\[
P_1(X)=X^6+X^3+X^2+X+1,
\]
\[
P_2(X)=X^6+X^5+X^2+1,
\]
\[
P_3(X)=X^6+X^5+X^3+X^2+X+1.
\]
A direct verification shows that
\[
P_1(X)\mid X^{15}-1,\quad
P_2(X)\mid X^{21}-1,\quad
P_3(X)\mid X^{31}-1,
\]
and
\[
\gcd(P_i(X),X^3+X+1)=1,\quad i=1,2,3.
\]
Therefore, if \(15\mid\ell\), or \(21\mid\ell\), or \(31\mid\ell\), then \(d=6t+1\).
\end{proof}
\begin{remark}
Under the hypotheses of Proposition~\ref{d12}, Proposition~\ref{prop:vdv-dimension} also gives the exact dimension in the following cases. If \(7\mid \ell\) and \(t\le \lfloor 2^{m/2}/6\rfloor\), then the code has parameters \([\,2^m-3t,\ 2^m-3t-3mt,\ 6t+1\,]\). If \(7\nmid \ell\) and \(t\le \lfloor(2^{m/2}+2)/6\rfloor\), then the code has parameters \([\,2^m,\ 2^m-3mt,\ 6t+1\,]\).
\end{remark}
\begin{example}
Table~\ref{tab:d12-examples} gives examples for Proposition~\ref{d12}. Over \(\mathbb F_{2^m}\), we take \(g(X)=X^{3t}+X^t+1\) and \(L=\{\alpha\in\mathbb F_{2^m}:g(\alpha)\ne0\}\), with \(\ell=(2^m-1)/t\). The parameters of \(\Gamma(L,g)\) were computed by Magma, and the listed minimum distances agree with Proposition~\ref{d12}.

\begin{table}[htbp]
\centering
\begingroup
\setlength{\tabcolsep}{5.5pt}
\renewcommand{\arraystretch}{1.12}
\caption{Parameters of \(\Gamma(L,g)\) for \(g(X)=X^{3t}+X^t+1\)}
\label{tab:d12-examples}
\begin{tabular}{c c c c c}
\toprule
\(m\) & \(t\) & \(\ell\) &  \(g(X)\) & Parameters \\
\midrule
\(4\) & \(1\) & \(15\) & \(X^3+X+1\) & \([16,4,7]\)\\
\(8\) & \(1\) & \(255\) & \(X^3+X+1\) & \([256,232,7]\)\\
\(6\) & \(1\) & \(63\) & \(X^3+X+1\) & \([61,43,7]\)\\
\(6\) & \(3\) & \(21\) & \(X^9+X^3+1\) & \([55,5,19]\)\\
\(5\) & \(1\) & \(31\) & \(X^3+X+1\) & \([32,17,7]\)\\
\(10\) & \(3\) & \(341\) & \(X^9+X^3+1\) & \([1024,934,19]\)\\
\bottomrule
\end{tabular}
\endgroup
\end{table}
    
\end{example}
\begin{remark}
For a fixed polynomial \(f(X)\), one may first use the conditions \(f(X)\mid P(X)+XP'(X)\) and \(\gcd(P(X),f(X))=1\) to determine possible auxiliary polynomials \(P(X)\), and then use \(P(X)\mid X^\ell-1\) to determine possible values of \(\ell\). Thus, the same method can systematically produce more infinite families of binary Goppa codes with determined minimum distance. For example, for \(f(X)=X^3+X^2+1\), namely \(g(X)=X^{3t}+X^{2t}+1\), one may search for polynomials \(P(X)\) satisfying the above conditions and obtain corresponding infinite families. Moreover, for the two polynomials \(f(X)=X^2+X+1\) and \(f(X)=X^3+X+1\) considered above, the possible choices of \(\ell\) need not be limited to those listed here. Further searches for auxiliary polynomials \(P(X)\) may produce additional admissible values of \(\ell\). Therefore, we list only several representative infinite families to illustrate the effectiveness of the method.
\end{remark}

\subsection{Composite Goppa Polynomials \(g(X)=A(X)h(\phi(X))\)}

\begin{theorem}\label{d2}
Let \(g(X)=A(X)h(\phi(X))\) be a separable polynomial, where \(A(X),h(X),\phi(X)\in\mathbb F_{2^m}[X]\), \(\deg\phi=\delta\ge1\), and \(\deg g=r\). Let
$L=\{\alpha\in\mathbb F_{2^m}:g(\alpha)\ne0\}.$
Assume further that \(A(X)\mid\phi'(X)\) and
\(\delta\mid 2r+1\). Set
$
N=\frac{2r+1}{\delta}.
$
Suppose that there exist pairwise distinct elements \(u_1,\dots,u_N\in\mathbb F_{2^m}\) satisfying the following three conditions:
\begin{enumerate}
\item\label{condition} for each \(i\), the equation \(\phi(X)=u_i\) has \(\delta\) distinct roots in \(\mathbb F_{2^m}\);
\item for each \(i\), \(h(u_i)\ne0\);
\item if \(F_0(X)=\prod_{i=1}^N(X-u_i)\), then \(h(X)\mid F_0'(X)\).
\end{enumerate}
Then the minimum distance of the binary Goppa code $\Gamma(L,g)$ equals its designed distance, i.e., 
$d=2r+1.$
\end{theorem}
\begin{proof}
Since \(g(X)\) is separable, by Lemma~\ref{41}, the binary Goppa code \(\Gamma(L,g)\) satisfies \(d\ge 2\deg g+1=2r+1\).
Thus, it remains to construct a codeword of weight \(2r+1\).

Since \(A(X)\mid \phi'(X)\), there exists \(B(X)\in \mathbb F_{2^m}[X]\) such that
$\phi'(X)=A(X)B(X).$
For each \(i\), write
\[
        S_i=\{\alpha\in \mathbb F_{2^m}:\phi(\alpha)=u_i\}.
\]
By condition (1), \(|S_i|=\delta\). Since the elements \(u_1,\ldots,u_N\) are pairwise
distinct, the sets \(S_i\) are pairwise disjoint. We now construct a subset \(S\subseteq L\)
of size \(2r+1\). Let
\[
        S=\bigcup_{i=1}^N S_i .
\]
Then$|S|=N\delta=2r+1.$

We first prove that \(S\subseteq L\).
Let \(\alpha\in S\). Then
\(\alpha\in S_i\) for some \(1\le i\le N\), and hence \(\phi(\alpha)=u_i\). Since \(\phi(X)-u_i\) has \(\alpha\) as a simple root, we have
\(\phi'(\alpha)\ne0\). As \(\phi'(\alpha)=A(\alpha)B(\alpha)\), it follows that
\(A(\alpha)\ne0\). By condition (2),
        $h(\phi(\alpha))=h(u_i)\ne0.$
Therefore        $g(\alpha)=A(\alpha)h(\phi(\alpha))\ne0,$
and hence \(\alpha\in L\). Thus \(S\subseteq L\).

Next, we prove that the binary vector whose nonzero positions correspond exactly to
the elements of \(S\) is a codeword of \(\Gamma(L,g)\). Let
\[
        \sigma_S(X)=\prod_{\beta\in S}(X-\beta).
\]
Let \(\lambda\) be the leading coefficient of \(\phi(X)\). Since \(\phi(X)-u_i\) has root set \(S_i\) and has no multiple roots, we have
\[
        \phi(X)-u_i=\lambda\prod_{\alpha\in S_i}(X-\alpha).
\]
Hence
\[
        \sigma_S(X)=\lambda^{-N}\prod_{i=1}^N(\phi(X)-u_i)
        =\lambda^{-N}F_0(\phi(X)).
\]
Thus
\[
        \sigma_S'(X)=\lambda^{-N}F_0'(\phi(X))\phi'(X).
\]
By condition (3), there exists \(Q(X)\in \mathbb F_{2^m}[X]\) such that
\[
        F_0'(X)=h(X)Q(X).
\]
Therefore
\[
\begin{aligned}
        \sigma_S'(X)
        &=\lambda^{-N}h(\phi(X))Q(\phi(X))\phi'(X)  \\
        &=\lambda^{-N}A(X)h(\phi(X))B(X)Q(\phi(X)) \\
        &=\lambda^{-N}g(X)B(X)Q(\phi(X)).
\end{aligned}
\]
Hence \(g(X)\mid \sigma_S'(X)\).

Writing \(S=\{\beta_1,\ldots,\beta_{2r+1}\}\), and using \(S\subseteq L\), we have
\(g(\beta_j)\ne0\) for all \(j\). Hence \(\gcd(\sigma_S(X),g(X))=1\). It follows that
\[
        \sum_{j=1}^{2r+1}\frac{1}{X-\beta_j}
        =
        \frac{\sigma_S'(X)}{\sigma_S(X)}
        \equiv 0 \pmod{g(X)}.
\]
Therefore, the binary vector supported exactly on \(S\) is a codeword of
\(\Gamma(L,g)\) of weight \(2r+1\). Hence \(d\le 2r+1\). Combining this with the lower
bound \(d\ge 2r+1\), we obtain \(d=2r+1\).
\end{proof}
\begin{remark}\label{rem:cubic-fiber}
Theorem~\ref{d2} gives a sufficient condition. For some special maps, the conditions in the theorem can be simplified. For example, take \(A(X)=X\) and \(\phi(X)=X^3\). Since \(\phi'(X)=X^2\), we have \(A(X)\mid\phi'(X)\). If \(m\) is even, then \(3\mid 2^m-1\), and the equation \(X^3=u\) has three distinct roots in \(\mathbb F_{2^m}\) if and only if \(u\) is a nonzero cubic element. Therefore, when the Goppa polynomial has the form
\[
g(X)=Xh(X^3),
\]
where \(\deg h=s\), condition~(\ref{condition}) in Theorem~\ref{d2} can be reduced to requiring \(2s+1\) suitable nonzero cubic elements in \(\mathbb F_{2^m}\). Thus, it remains to find \(2s+1\) cubic elements satisfying the algebraic condition arising from \(h(X)\mid F_0'(X)\) in Theorem~\ref{d2}. In some special cases, the existence of such elements can be determined using exponential sums.
\end{remark}
We now apply Theorem~\ref{d2} to give an explicit infinite family. In this case, the theorem reduces the problem of determining whether the designed distance is attained to the solvability of certain equations over finite fields. We verify the required solvability using exponential sums. To facilitate the computation of the relevant exponential sums, we first give the following two lemmas.



For an additive character \(\psi\) and a multiplicative character \(\chi\) of
\(\mathbb F_{2^m}\), let
\[
        G_{\mathbb F_{2^m}}(\chi,\psi)
        =
        \sum_{x\in \mathbb F_{2^m}^*}\chi(x)\psi(x)
\]
denote the corresponding Gauss sum.

\begin{lemma} [Davenport--Hasse lifting theorem \cite{28}]\label{DH}
Let \(\psi_0\) be an additive character and \(\chi_0\) a multiplicative character of
\(\mathbb F_4\), not both trivial. Suppose that \(\psi_0\) and \(\chi_0\) are lifted to
characters \(\psi\) and \(\chi\), respectively, of \(\mathbb F_{2^m}\), where \(m\) is even.
Then
\[
        G_{\mathbb F_{2^m}}(\chi,\psi)
        =
        (-1)^{m/2-1}G_{\mathbb F_4}(\chi_0,\psi_0)^{m/2}.
\]
Here the lifted characters are defined by
\[
        \psi(x)=\psi_0\bigl(\operatorname{Tr}_{\mathbb F_{2^m}/\mathbb F_4}(x)\bigr),
        \quad
        \chi(x)=\chi_0\bigl(\operatorname{N}_{\mathbb F_{2^m}/\mathbb F_4}(x)\bigr).
\]
\end{lemma}

\begin{lemma}\label{lem:periods-order-3}
Let \(m\) be a positive even integer, and put \(\mu=2^{m/2}\) and \(\varepsilon=(-1)^{m/2-1}\). Let \(w\) be a primitive element of \(\mathbb F_{2^m}^*\), and set \(T_0=\{z^3:z\in\mathbb F_{2^m}^*\}\), \(T_1=wT_0\), and \(T_2=w^2T_0\). Define
\[
\eta_i=\sum_{x\in T_i}(-1)^{\operatorname{Tr}_{\mathbb F_{2^m}/\mathbb F_2}(x)},\quad i=0,1,2.
\]
Then
\[
\eta_0=\frac{-1+2\varepsilon\mu}{3},\quad
\eta_1=\eta_2=\frac{-1-\varepsilon\mu}{3}.
\]
\end{lemma}

\begin{proof}
Since \(m\) is even, we have \(3\mid 2^m-1\). Hence \(T_0=\{z^3:z\in
\mathbb F_{2^m}^*\}\) is a subgroup of \(\mathbb F_{2^m}^*\) of index \(3\), and
\(T_1,T_2\) are the other two multiplicative cosets modulo \(T_0\). We may view \(\mathbb F_{2^m}\) as an extension of \(\mathbb F_4\)
of degree \(m/2\). 

Let \(\mathbb F_4^*=\langle \rho\rangle\), where
\(\rho^2+\rho+1=0\). Let \(\zeta_3\) be a complex primitive third root of unity,
and let \(\chi_0\) be the multiplicative character of order three on
\(\mathbb F_4\) such that \(\chi_0(\rho)=\zeta_3\). 
Also let
        $\psi_0(x)=(-1)^{\operatorname{Tr}_{\mathbb F_4/\mathbb F_2}(x)}$
be the standard additive character of \(\mathbb F_4\). Since
        $\operatorname{Tr}_{\mathbb F_4/\mathbb F_2}(1)=0$ and 
        $\operatorname{Tr}_{\mathbb F_4/\mathbb F_2}(\rho)
        =
        \operatorname{Tr}_{\mathbb F_4/\mathbb F_2}(\rho^2)=1$,
we have
\[
\begin{aligned}
        G_{\mathbb F_4}(\chi_0,\psi_0)
        =
        \sum_{x\in\mathbb F_4^*}\chi_0(x)\psi_0(x)        
        =1-\zeta_3-\zeta_3^2
        =2.
\end{aligned}
\]

Let
        $\chi=\chi_0\circ \operatorname{N}_{\mathbb F_{2^m}/\mathbb F_4}.$
Then \(\chi\) is the lift of \(\chi_0\) to \(\mathbb F_{2^m}\). Moreover, the
kernel of \(\chi\) is \(T_0\). Indeed, if \(w\) is a primitive element of
\(\mathbb F_{2^m}^*\), then
\[
        \operatorname{N}_{\mathbb F_{2^m}/\mathbb F_4}(w)
        =
        w^{(2^m-1)/3}
\]
has order three, and hence \(\chi(w)\) is a primitive third root of unity. Thus
\(\ker\chi=\langle w^3\rangle=T_0\). 

Let $\psi(x)=(-1)^{\operatorname{Tr}_{\mathbb F_{2^m}/\mathbb F_2}(x)},$
then
\[
        \eta_i=\sum_{x\in T_i}\psi(x),\quad i=0,1,2.
\]
By the transitivity of trace,
\[
        \psi_0\bigl(\operatorname{Tr}_{\mathbb F_{2^m}/\mathbb F_4}(x)\bigr)
        =
        (-1)^{\operatorname{Tr}_{\mathbb F_{2^m}/\mathbb F_2}(x)}
        =
        \psi(x),
\]
so \(\psi\) is the lift of \(\psi_0\). Therefore, by Lemma \ref{DH},
\[
        G_{\mathbb F_{2^m}}(\chi,\psi)
        =
        (-1)^{m/2-1}G_{\mathbb F_4}(\chi_0,\psi_0)^{m/2}
        =
        \varepsilon\mu .
\]

Since \(\psi\) is nontrivial, we have
\[
        \sum_{x\in\mathbb F_{2^m}^*}\psi(x)=-1.
\]
Hence $\eta_0+\eta_1+\eta_2=-1.$
Since the field has characteristic \(2\), the square map is the Frobenius automorphism.
Moreover,                                                    
\[
        \operatorname{Tr}_{\mathbb F_{2^m}/\mathbb F_2}(x^2)
        =
        \operatorname{Tr}_{\mathbb F_{2^m}/\mathbb F_2}(x),
\]
and therefore \(\psi(x^2)=\psi(x)\). The square map fixes \(T_0\) and interchanges
\(T_1\) and \(T_2\). Hence \(\eta_1=\eta_2\).

Finally, \(\chi\) is equal to \(1\) on \(T_0\), and it takes the two values
\(\zeta_3\) and \(\zeta_3^2\) on \(T_1\) and \(T_2\), in some order. Since
\(\eta_1=\eta_2\), the order is irrelevant. Thus
\[
\begin{aligned}
        G_{\mathbb F_{2^m}}(\chi,\psi)
        =
        \sum_{x\in\mathbb F_{2^m}^*}\chi(x)\psi(x)  
        =
        \eta_0+\zeta_3\eta_1+\zeta_3^2\eta_2
        =
        \eta_0-\eta_1 .
\end{aligned}
\]
Solving
\[
        \eta_0+2\eta_1=-1,\quad \eta_0-\eta_1=\varepsilon\mu,
\]
we obtain
\[
        \eta_0=\frac{-1+2\varepsilon\mu}{3},
        \quad
        \eta_1=\eta_2=\frac{-1-\varepsilon\mu}{3}.
\]
\end{proof}
\begin{theorem}\label{thm:x4cx-distance}
Let \(m\ge6\) be even, and let \(c\in\mathbb F_{2^m}^*\). Put
\[
g(X)=X^4+cX,
\quad
L=\{\alpha\in\mathbb F_{2^m}:g(\alpha)\ne0\}.
\]
Then the binary separable Goppa code \(\Gamma(L,g)\) has minimum distance
$d=9.$
\end{theorem}

\begin{proof}
Let \(\mu=2^{m/2}\) and \(\varepsilon=(-1)^{m/2-1}\). Let \(w\) be a primitive element of \(\mathbb F_{2^m}^*\), and put \(T_0=\{z^3:z\in\mathbb F_{2^m}^*\}\), \(T_1=wT_0\), and \(T_2=w^2T_0\). Suppose that \(c\in T_k\), where \(k\in\{0,1,2\}\).

Let \(h(X)=X+c\). Then \(g(X)=Xh(X^3)\) and \(\deg h=1\). By Remark~\ref{rem:cubic-fiber}, it is enough to find three pairwise distinct nonzero cubic elements \(u_1,u_2,u_3\in\mathbb F_{2^m}\) such that \(h(u_i)\ne0\) for \(i=1,2,3\), and, for \(F_0(X)=\prod_{i=1}^3(X-u_i)\), one has \(h(X)\mid F_0'(X)\). Since the cubic elements in \(\mathbb F_{2^m}^*\) are precisely the elements of \(T_0\), this amounts to finding pairwise distinct \(u_1,u_2,u_3\in T_0\) with \(u_i\ne c\) such that \(X+c\mid F_0'(X)\).

Now
\[
F_0'(X)=X^2+u_1u_2+u_1u_3+u_2u_3.
\]
Hence \(X+c\mid F_0'(X)\) is equivalent to
\[
u_1u_2+u_1u_3+u_2u_3=c^2.
\]
Therefore, if we prove that this equation has pairwise distinct solutions \(u_1,u_2,u_3\in T_0\) with \(u_i\ne c\), then all the conditions of Theorem~\ref{d2} are satisfied.

Consider the equation
\[
u_1u_2+u_1u_3+u_2u_3=c^2,\quad u_1,u_2,u_3\in T_0.
\]
Let
\[
N_k=\#\{(a,b,d)\in T_k^3:a+b+d=1\}.
\]
Since \(c\in T_k\), we have \(c^2\in T_{2k}\). Hence, for
\(u_1,u_2,u_3\in T_0\), each of
\[
\frac{u_1u_2}{c^2},\quad
\frac{u_1u_3}{c^2},\quad
\frac{u_2u_3}{c^2}
\]
belongs to \(T_{-2k}=T_k\), where the subscripts are read modulo \(3\).
The map
\[
(u_1,u_2,u_3)\longmapsto
\left(\frac{u_1u_2}{c^2},\frac{u_1u_3}{c^2},\frac{u_2u_3}{c^2}\right)
\]
is a bijection from the ordered solutions of the above equation to \(\{(a,b,d)\in T_k^3:a+b+d=1\}\). Indeed, the inverse is uniquely determined by
\[
u_1^2=c^2\frac{ab}{d},\quad
u_2^2=c^2\frac{ad}{b},\quad
u_3^2=c^2\frac{bd}{a}.
\]
Conversely, if \(a,b,d\in T_k\), then
\[
\frac{ab}{d},\quad \frac{ad}{b},\quad \frac{bd}{a}
\]
all belong to \(T_k\). Hence the three right-hand sides in the displayed
equations above belong to \(T_{2k}T_k=T_0\). Since
\(\lvert T_0\rvert=(2^m-1)/3\) is odd, the square map is a bijection on
\(T_0\). Therefore, these equations uniquely determine
\(u_1,u_2,u_3\in T_0\).

We compute \(N_k\). By the orthogonality of additive characters,
\[
N_k=\frac{1}{2^m}\sum_{\lambda\in\mathbb F_{2^m}}\psi(\lambda)
\left(\sum_{a\in T_k}\psi(\lambda a)\right)^3 .
\]
The contribution of \(\lambda=0\) is \(|T_0|^3\). If \(\lambda\in T_j\), then \(\sum_{a\in T_k}\psi(\lambda a)=\eta_{j+k}\). Hence
\[
N_k=\frac{1}{2^m}\left(|T_0|^3+\sum_{j=0}^2\eta_j\eta_{j+k}^3\right),
\]
where the subscripts are read modulo \(3\). By Lemma~\ref{lem:periods-order-3}, a direct calculation gives
\[
N_0=\frac{\mu^4+3\mu^2+15-8\varepsilon\mu}{27},
\]
and
\[
N_1=N_2=\frac{\mu^4-6\mu^2+6+\varepsilon\mu}{27}.
\]

If \(c\notin T_0\), then \(k=1\) or \(2\), and \(N_k=N_1=N_2>0\) because \(\mu\ge8\). Hence there exist \(u_1,u_2,u_3\in T_0\) satisfying the above equation. Since \(c\notin T_0\), we automatically have \(u_i\ne c\). Moreover, if \(u_i=u_j\), then the equation gives \(u_i^2=c^2\), and hence \(u_i=c\), a contradiction. Therefore \(u_1,u_2,u_3\) are pairwise distinct.

If \(c\in T_0\), then \(k=0\). We must exclude the cases in which some \(u_i=c\) or two of the \(u_i\)'s are equal. If some \(u_i=c\), then the equation forces at least one of the other two variables to be equal to \(c\). If \(u_i=u_j\), then the same equation gives \(u_i=c\). Thus the bad solutions are precisely the ordered triples in which at least two coordinates are equal to \(c\). Their number is \(3|T_0|-2=2^m-3=\mu^2-3\). Hence the number of good solutions is at least
\[
N_0-(2^m-3)=\frac{\mu^4-24\mu^2+96-8\varepsilon\mu}{27}>0,
\]
where the last inequality follows directly from \(\mu\ge8\). Therefore, also in the case \(c\in T_0\), there exist pairwise distinct \(u_1,u_2,u_3\in T_0\) such that \(u_i\ne c\) and
\[
u_1u_2+u_1u_3+u_2u_3=c^2.
\]

By the reduction at the beginning of the proof, these elements satisfy all the required conditions in Theorem~\ref{d2}. Hence \(d=2\deg g+1=9\).
\end{proof}
\begin{remark}
For the family in Theorem~\ref{thm:x4cx-distance}, Proposition~\ref{prop:vdv-dimension} gives the exact dimension \(k=n-4m\). Hence, if \(c\) is a cubic element of \(\mathbb F_{2^m}^*\), then the code has parameters \([\,2^m-4,\ 2^m-4m-4,\ 9\,]\); otherwise, it has parameters \([\,2^m-1,\ 2^m-4m-1,\ 9\,]\).
\end{remark}

\begin{example}
Table~\ref{tab:x4cx-examples} lists several concrete parameters of \(\Gamma(L,g)\) computed by Magma. In the table, \(L=\{\alpha\in\mathbb F_{2^m}:g(\alpha)\ne0\}\), and the minimum distances agree with Theorem~\ref{thm:x4cx-distance}.

\begin{table}[htbp]
\centering
\setlength{\tabcolsep}{6pt}
\renewcommand{\arraystretch}{1.12}
\caption{Parameters of \(\Gamma(L,g)\) for \(g(X)=X^4+cX\)}
\label{tab:x4cx-examples}
\begin{tabular}{c c c c}
\toprule
\(\mathbb F_{2^m}\) & Defining polynomial &  \(g(X)\) & Parameters \([n,k,d]\)\\
\midrule
\(\mathbb F_{2^6}=\mathbb F_2(w)\)
& \(w^6+w^4+w^3+w+1=0\)
& \(X^4+wX\)
& \([63,39,9]\)\\
\(\mathbb F_{2^6}\)
& --
& \(X^4+X\)
& \([60,36,9]\)\\
\(\mathbb F_{2^8}\)
& --
& \(X^4+X\)
& \([252,220,9]\)\\
\(\mathbb F_{2^8}=\mathbb F_2(w)\)
& \(w^8+w^4+w^3+w^2+1=0\)
& \(X^4+wX\)
& \([255,223,9]\)\\
\bottomrule
\end{tabular}
\end{table}
\end{example}
\begin{remark}
In Theorem~\ref{thm:x4cx-distance}, when \(c=1\), the corresponding Goppa codes admit an \(A_4\) automorphism group, are quasi-cyclic, and have determined parameters. This will be proved in Theorem~\ref{main}.
\end{remark}

In Theorem~\ref{d2}, besides \(\phi(X)=X^3\), one may also choose other polynomial maps. For example, take \(A(X)=X^2+bX+d\), where \(b,d\in\mathbb F_{2^m}\) and \(b\ne0\), and set \(\phi(X)=X^5+b^2X^3+d^2X\). Since \(\phi'(X)=A(X)^2\), we have \(A(X)\mid\phi'(X)\), so the main condition in Theorem~\ref{d2} is satisfied. We give a concrete example below.

\begin{example}
Let \(\mathbb F_{2^8}=\mathbb F_2(w)\), where \(w^8+w^4+w^3+w^2+1=0\). Take \(A(X)=X^2+X+1\), \(\phi(X)=X^5+X^3+X\), and \(h(X)=X+w^{195}\). Then
\[
g(X)=A(X)h(\phi(X))=(X^2+X+1)(X^5+X^3+X+w^{195}).
\]
Let \(U=\{1,w^{142},w^{29}\}\). Computations show that, for every \(u\in U\), the equation \(X^5+X^3+X=u\) has five distinct roots in \(\mathbb F_{2^8}\). Moreover, if
\[
F_0(X)=(X-1)(X-w^{142})(X-w^{29}),
\]
then \(h(X)\mid F_0'(X)\). Therefore, Theorem~\ref{d2} gives \(d=15\), where \(L=\{\alpha\in\mathbb F_{2^8}:g(\alpha)\ne0\}\). Magma gives the parameters
$[253,197,15].$
\end{example}

\section{Binary Goppa codes and their related codes with \(A_4\) or \(A_5\) automorphism groups}\label{sec4}
In this section, we study binary Goppa codes and their related codes with alternating automorphism groups. We first construct binary Goppa codes, expurgated Goppa codes and extended Goppa codes with \(A_4\) automorphism groups by using the orbit decompositions of \(A_4\)-subgroups of \(PGL(2,2^m)\) on \(\overline{\mathbb F}_{2^m}\). We then determine the parameters of a representative \(A_4\)-invariant family by applying the minimum-distance criterion developed in Section~\ref{sec3}, and derive a dimension result for another related family. Finally, we extend the construction to \(A_5\) automorphism groups.

We first recall the necessary and sufficient condition for \(PGL(2,2^m)\) to contain subgroups isomorphic to \(A_4\) and \(A_5\). 

\begin{lemma}\cite{Suzuki1982}
Let \(L=PGL(2,2^m)\). Then \(L\) contains subgroups isomorphic to \(A_4\) and \(A_5\) if and only if \(m\) is even.
\end{lemma}

Therefore, we assume that \(m\ge 6\) is even in this section.
\subsection{Construction with \(A_4\) Automorphism Groups}

To obtain a matrix representation of \(A_4\), we first state the following conjugacy result for subgroups isomorphic to \(A_4\) in \(PGL(2,2^m)\).

\begin{lemma}\cite{Suzuki1982}\label{7}
Let \(PGL(2,2^m)\) contain a subgroup \(H\) isomorphic to \(A_4\). Then \(H\) is conjugate to a subgroup of the standard subgroup \(PGL(2,4)\). In particular, any two subgroups of \(PGL(2,2^m)\) that are isomorphic to \(A_4\) are conjugate to each other.
\end{lemma}

By this lemma, it remains to choose a standard representative of the conjugacy class of \(A_4\)-subgroups in \(PGL(2,2^m)\). Let \(\omega\in\mathbb F_4\) be a primitive element satisfying
\(
\omega^2+\omega+1=0.
\)
Define
\[
A=
\begin{pmatrix}
\omega & 0\\
0 & \omega^2
\end{pmatrix},
\quad
B=
\begin{pmatrix}
1 & 1\\
0 & 1
\end{pmatrix}.
\]

\begin{proposition}
Define
\(
H_0=\langle A,B\rangle\le PGL(2,2^m).
\)
Then
\(
H_0\cong A_4.
\)
Moreover,
\[
\{H\le PGL(2,2^m):H\cong A_4\}
=
\{PH_0P^{-1}:P\in PGL(2,2^m)\}.
\]
\end{proposition}
\begin{proof}
Let
$\sigma=(1\,2\,3)$, $\tau=(1\,2)(3\,4).$
Then \(A_4=\langle\sigma,\tau\rangle\), where \(\sigma^3=\tau^2=(\sigma\tau)^3=1\). On the other hand, a direct calculation gives
\[
A^3=B^2=(AB)^3=1.
\]
Thus the assignment \(\sigma\mapsto A\) and \(\tau\mapsto B\) defines a homomorphism from \(A_4\) onto \(H_0=\langle A,B\rangle\).

It remains to compare the orders. The powers of \(A\) induce the maps \(z\mapsto \alpha z\), where \(\alpha\in\mathbb F_4^*\), and the elements \(1,B,ABA^{-1},A^2BA^{-2}\) give all translations \(z\mapsto z+\beta\), where \(\beta\in\mathbb F_4\). Hence \(H_0\) contains all transformations \(z\mapsto \alpha z+\beta\) with \(\alpha\in\mathbb F_4^*\) and \(\beta\in\mathbb F_4\), and therefore \(|H_0|\ge 12\). Since \(H_0\) is a homomorphic image of \(A_4\), we also have \(|H_0|\le |A_4|=12\). Thus \(|H_0|=12\), and the homomorphism is an isomorphism. Hence \(H_0\cong A_4\).

The second assertion follows immediately from Lemma~\ref{7}.
\end{proof}
We now use the above representation to construct Goppa codes and related codes with \(A_4\) automorphism groups.
Let
$H_0=\langle A,B\rangle\cong A_4$
be the standard subgroup of \(PGL(2,2^m)\) defined above. For \(P\in PGL(2,2^m)\), set
\[
G=PH_0P^{-1}\le PGL(2,2^m).
\]
Next, we distinguish two cases: the affine case and the non-affine case.

\textbf{\textit{Case 1 (Affine case).}}
Let
\[
P=
\begin{pmatrix}
a & b\\
0 & 1
\end{pmatrix}
\in PGL(2,2^m).
\]
The following theorem describes the orbit decomposition of the conjugate subgroup \(G=PH_0P^{-1}\) in this case.
\begin{proposition}\label{orbit1}
The action of \(G=PH_0P^{-1}\) on
$ \overline{\mathbb F}_{2^m}=\mathbb F_{2^m}\cup\{\infty\}$
has the following orbit decomposition:
\[
 \overline{\mathbb F}_{2^m}
=
O_0\cup O_1\cup \bigcup_{i=2}^{\frac{2^m+8}{12}} O_i,
\]
where
$O_0=\{\infty\},
O_1=\{b,a+b,a\omega+b,a\omega^2+b\},$
and
$
|O_i|=12,~  2\le i\le \frac{2^m+8}{12}.
$
\end{proposition}
\begin{proof}
	We first consider the action of \(H_0\) on $\overline{\mathbb F}_{2^m}$. Since all matrices in \(H_0\) are upper triangular and have entries in \(\mathbb F_4\), we have
	\[
	O_0=\{\infty\},
	\quad
	O_1=\{0,1,\omega,\omega^2\}.
	\]
	For any \(\gamma\in\mathbb F_{2^m}\setminus \mathbb F_4\), no nonidentity element of \(H_0\) fixes \(\gamma\). Hence, by the orbit-stabilizer theorem,
	\[
	|H_0(\gamma)|
	=
	\frac{|H_0|}{|(H_0)_\gamma|}
	=
	|H_0|
	=
	12.
	\]
	Therefore, all remaining orbits have length \(12\). Applying the conjugation by \(P\), we obtain the stated orbit decomposition for \(G=PH_0P^{-1}\).
\end{proof}
Using the long orbits in this decomposition as the support, we obtain Goppa codes and expurgated Goppa codes with  \(A_4\) automorphism groups.
\begin{theorem}\label{5}
Let \(L=\bigcup_{i=2}^{s}O_i\), where \(O_2,\ldots,O_s\), \(2\le s\le (2^m+8)/12\), are the orbits defined in Proposition~\ref{orbit1}. 
Let \(\eta_1,\ldots,\eta_t\) be elements in an extension field of
\(\mathbb F_{2^m}\), and let \(e_1,\ldots,e_t\) be positive integers.
Assume that \(G(\eta_j)\cap L=\varnothing\) for \(1\leq j\leq t\), and that
\[
g(X):=
\prod_{j=1}^t
\left(\prod_{\alpha\in G(\eta_j)}(X-\alpha)\right)^{e_j}
\]
belongs to \(\mathbb F_{2^m}[X]\).
Then the binary Goppa code $\Gamma(L,g)$ and the expurgated Goppa
code $\widetilde{\Gamma}(L,g)$ have an automorphism group
isomorphic to $A_4$.
\end{theorem}
\begin{proof}
	Let
	\[
	G=\langle A_P,B_P\rangle,
	\quad
	A_P=PAP^{-1},\quad B_P=PBP^{-1}.
	\]
	By the definition of \(g(X)\), the set of roots of \(g(X)\) is stable under the action of \(G\). Hence there exist nonzero constants \(\gamma_1,\gamma_2\) such that
	\[
	g(A_P(\zeta))=\gamma_1g(\zeta),
	\quad
	g(B_P(\zeta))=\gamma_2g(\zeta).
	\]
	
	For each long orbit \(O_i\), write
	\[
	\begin{aligned}
	O_i=\{&
	\zeta_i,A_P(\zeta_i),A_P^2(\zeta_i),B_P(\zeta_i),
	A_PB_P(\zeta_i),A_P^2B_P(\zeta_i),A_PB_PA_P^2(\zeta_i),\\
	&B_PA_PB_P(\zeta_i),A_PB_PA_PB_P(\zeta_i),B_PA_P(\zeta_i),
	A_PB_PA_P(\zeta_i),A_P^2B_PA_P(\zeta_i)\}.
	\end{aligned}
	\]
	Then the action of \(A_P\) on \(L\) induces the permutation
	\[
	\psi_A
	=
	\prod_{i=2}^{s}
	(1^{(i)},2^{(i)},3^{(i)})
	(4^{(i)},5^{(i)},6^{(i)})
	(7^{(i)},8^{(i)},9^{(i)})
	(10^{(i)},11^{(i)},12^{(i)}),
	\]
	and the action of \(B_P\) induces the permutation
	\[
	\psi_B
	=
	\prod_{i=2}^{s}
	(1^{(i)},4^{(i)})
	(2^{(i)},10^{(i)})
	(3^{(i)},9^{(i)})
	(5^{(i)},8^{(i)})
	(6^{(i)},11^{(i)})
	(7^{(i)},12^{(i)}).
	\]
	Thus,
	\[
	\langle \psi_A,\psi_B\rangle\cong A_4.
	\]
    Moreover, since \(g(A_P(\zeta))\) and \(g(B_P(\zeta))\) are scalar multiples of \(g(\zeta)\), the parity-check matrices obtained after applying the induced coordinate permutations have the same row spaces as the original parity-check matrices of \(\Gamma(L,g)\) and \(\widetilde{\Gamma}(L,g)\). Hence these permutations preserve both \(\Gamma(L,g)\) and \(\widetilde{\Gamma}(L,g)\). Therefore, the two codes have an automorphism group isomorphic
to $A_4$.
\end{proof}
\begin{remark}
Since the automorphism group of each code constructed in Theorem~\ref{5} contains a subgroup isomorphic to \(A_4\), and \(A_4\) contains cyclic subgroups isomorphic to \(\mathbb Z/3\mathbb Z\) and \(\mathbb Z/2\mathbb Z\), the codes \(\Gamma(L,g)\) and \(\widetilde{\Gamma}(L,g)\) are binary quasi-cyclic Goppa codes of length \(12(s-1)\).
\end{remark}

We next give examples over \(\mathbb{F}_{2^6}\) arising from the
construction above.
\begin{example}
Let
$f(X)=X^6+X^4+X^3+X+1,$ 
and let \(w\) be a root of \(f(X)\) over \(\mathbb F_2\). Then
$\mathbb F_2(w)\cong \mathbb F_{2^6}.$

We first list the orbit decompositions for several choices of \(P\) in Case 1, see Table~\ref{biao1}. For each fixed \(P\), the symbols \(L_1,\ldots,L_5\) denote the long orbits of length \(12\) listed in the last column of Table~\ref{biao1}. Using these orbit labels, Table~\ref{biao2} lists the parameters of the Goppa codes and the corresponding expurgated Goppa codes obtained from the construction. These codes have an \(A_4\) automorphism subgroup and are quasi-cyclic.

\begingroup
\scriptsize
\setlength{\tabcolsep}{2.2pt}
\renewcommand{\arraystretch}{1.12}

\begin{longtable}{C{0.8cm} C{1.9cm} C{2.7cm} P{9.0cm}}
\caption{Orbit decompositions in the affine case over \(\mathbb F_{2^6}\)}
\label{biao1}\\
\toprule
\multicolumn{1}{c}{No.}
& \multicolumn{1}{c}{Matrix \(P\)}
& \multicolumn{1}{c}{Short orbits}
& \multicolumn{1}{c}{Long orbits of length \(12\)}\\
\midrule
\endfirsthead

\toprule
\multicolumn{1}{c}{No.}
& \multicolumn{1}{c}{Matrix \(P\)}
& \multicolumn{1}{c}{Short  orbits}
& \multicolumn{1}{c}{Long orbits of length \(12\)}\\
\midrule
\endhead

\multirow{5}{*}{1}
& \multirow{5}{*}{\(I\)}
& \multirow{5}{=}{\makecell[l]{fixed point: \(\infty\)\\short orbit:\\ \((0,1,w^{21},w^{42})\)}}
& \(L_1=(w,w^{43},w^{22},w^{56},w^{35},w^{14},w^{25},w^4,w^{46},w^{37},w^{16},w^{58})\)\\
&&& \(L_2=(w^2,w^{44},w^{23},w^{49},w^{28},w^7,w^{53},w^{32},w^{11},w^{29},w^8,w^{50})\)\\
&&& \(L_3=(w^3,w^{45},w^{24},w^{13},w^{55},w^{34},w^{20},w^{62},w^{41},w^{36},w^{15},w^{57})\)\\
&&& \(L_4=(w^5,w^{47},w^{26},w^{30},w^9,w^{51},w^{48},w^{27},w^6,w^{61},w^{40},w^{19})\)\\
&&& \(L_5=(w^{10},w^{52},w^{31},w^{60},w^{39},w^{18},w^{38},w^{17},w^{59},w^{12},w^{54},w^{33})\)\\
\midrule

\multirow{5}{*}{2}
& \multirow{5}{*}{\(\begin{pmatrix}1&w^3\\0&1\end{pmatrix}\)}
& \multirow{5}{=}{\makecell[l]{fixed point: \(\infty\)\\short orbit:\\ \((w^3,w^{13},w^{20},w^{57})\)}}
& \(L_1=(0,w^{24},w^{45},1,w^{15},w^{62},w^{42},w^{34},w^{36},w^{41},w^{55},w^{21})\)\\
&&& \(L_2=(w,w^9,w^{33},w^{56},w^{47},w^{10},w^{25},w^{61},w^{38},w^{27},w^{60},w^{58})\)\\
&&& \(L_3=(w^2,w^{18},w^6,w^{49},w^{59},w^{40},w^{53},w^{31},w^{26},w^{54},w^{51},w^{50})\)\\
&&& \(L_4=(w^4,w^7,w^{19},w^{35},w^{11},w^5,w^{37},w^{23},w^{48},w^8,w^{30},w^{43})\)\\
&&& \(L_5=(w^{12},w^{29},w^{22},w^{52},w^{32},w^{14},w^{17},w^{28},w^{46},w^{44},w^{16},w^{39})\)\\
\midrule

\multirow{5}{*}{3}
& \multirow{5}{*}{\(\begin{pmatrix}w^2&w^5\\0&1\end{pmatrix}\)}
& \multirow{5}{=}{\makecell[l]{fixed point: \(\infty\)\\short orbit:\\ \((w^5,w^{15},w^{22},w^{59})\)}}
& \(L_1=(0,w^{26},w^{47},w^2,w^{17},w,w^{44},w^{36},w^{38},w^{43},w^{57},w^{23})\)\\
&&& \(L_2=(1,w^{40},w^{27},w^{49},w^{12},w^{58},w^{29},w^{62},w^{60},w^{35},w^3,w^{11})\)\\
&&& \(L_3=(w^4,w^{20},w^8,w^{51},w^{61},w^{42},w^{55},w^{33},w^{28},w^{56},w^{53},w^{52})\)\\
&&& \(L_4=(w^6,w^9,w^{21},w^{37},w^{13},w^7,w^{39},w^{25},w^{50},w^{10},w^{32},w^{45})\)\\
&&& \(L_5=(w^{14},w^{31},w^{24},w^{54},w^{34},w^{16},w^{19},w^{30},w^{48},w^{46},w^{18},w^{41})\)\\
\bottomrule
\end{longtable}
\endgroup

\begingroup
\scriptsize
\setlength{\tabcolsep}{2.2pt}
\renewcommand{\arraystretch}{1.16}

\begin{longtable}{C{1.5cm} P{3cm} C{6.4cm} C{1.55cm} C{1.45cm}}
\caption{Parameters of Goppa codes and related codes with \(A_4\) automorphism groups in the affine case over \(\mathbb F_{2^6}\)}
\label{biao2}\\

\toprule
\multicolumn{1}{c}{No.}
& \multicolumn{1}{c}{Support}
& \multicolumn{1}{c}{\(g(X)\)}
& \multicolumn{1}{c}{\(\Gamma(L,g)\)}
& \multicolumn{1}{c}{\(\widetilde{\Gamma}(L,g)\)}\\
\midrule
\endfirsthead

\caption[]{Parameters of Goppa codes and related codes with \(A_4\) automorphism groups in the affine case over \(\mathbb F_{2^6}\) (continued)}\\
\toprule
\multicolumn{1}{c}{No.}
& \multicolumn{1}{c}{Support}
& \multicolumn{1}{c}{\(g(X)\)}
& \multicolumn{1}{c}{\(\Gamma(L,g)\)}
& \multicolumn{1}{c}{\(\widetilde{\Gamma}(L,g)\)}\\
\midrule
\endhead

\multirow{5}{*}{1}
& \(L_1\cup L_2\)
& \(X^4+X\)
& \([24,4,12]\)
& \([24,4,12]\)\\

& \(L_1\cup L_2\cup L_3\)
& \(X^4+X\)
& \([36,12,9]\)
& \([36,11,10]\)\\

& \(L_1\cup L_2\cup L_3\cup L_4\)
& \(X^4+X\)
& \([48,24,9]\)
& \([48,23,10]\)\\

& \(L_1\cup\cdots\cup L_5\)
& \(X^4+X\)
& \([60,36,9]\)
& \([60,35,10]\)\\

& \(L_1\cup\cdots\cup L_5\)
& \(X^{12}+X^9+X^6+X^3+w^9\)
& \([60,4,27]\)
& \([60,3,34]\)\\
\midrule

\multirow{5}{*}{2}
& \(L_1\cup L_2\)
& \(X^4+X+w^{30}\)
& \([24,4,12]\)
& \([24,4,12]\)\\

& \(L_1\cup L_2\cup L_3\)
& \(X^4+X+w^{30}\)
& \([36,12,9]\)
& \([36,11,10]\)\\

& \(L_1\cup L_2\cup L_3\cup L_4\)
& \(X^4+X+w^{30}\)
& \([48,24,9]\)
& \([48,23,10]\)\\

& \(L_1\cup\cdots\cup L_5\)
& \(X^4+X+w^{30}\)
& \([60,36,9]\)
& \([60,35,10]\)\\

& \(L_1\cup\cdots\cup L_5\)
& \makecell[c]{\(X^{12}+X^9+w^{30}X^8+X^6+w^{60}X^4\)\\
\(+X^3+w^{30}X^2+w^{60}X+1\)}
& \([60,4,27]\)
& \([60,3,34]\)\\
\midrule

\multirow{5}{*}{3}
& \(L_1\cup L_2\)
& \(X^4+w^6X+w^{38}\)
& \([24,4,12]\)
& \([24,4,12]\)\\

& \(L_1\cup L_2\cup L_3\)
& \(X^4+w^6X+w^{38}\)
& \([36,12,9]\)
& \([36,11,10]\)\\

& \(L_1\cup L_2\cup L_3\cup L_4\)
& \(X^4+w^6X+w^{38}\)
& \([48,24,9]\)
& \([48,23,10]\)\\

& \(L_1\cup\cdots\cup L_5\)
& \(X^4+w^6X+w^{38}\)
& \([60,36,9]\)
& \([60,35,10]\)\\

& \(L_1\cup\cdots\cup L_5\)
& \makecell[c]{\(X^{12}+w^6X^9+w^{38}X^8+w^{12}X^6+w^{13}X^4\)\\
\(+w^{18}X^3+w^{50}X^2+w^{19}X+w^{15}\)}
& \([60,4,27]\)
& \([60,3,34]\)\\

\bottomrule
\end{longtable}

\endgroup

\end{example}
We next consider the case where \(P\) induces a non-affine transformation on \(\overline{\mathbb F}_{2^m}\).

\textbf{\textit{Case 2 (Non-affine case).}}
Let
\[
P=
\begin{pmatrix}
a & b\\
1 & d
\end{pmatrix}
\in PGL(2,2^m).
\]
The following theorem describes the orbit decomposition of the conjugate subgroup \(G=PH_0P^{-1}\) in this case.

\begin{proposition}\label{4}
	The action of \(G=PH_0P^{-1}\) on $\overline{\mathbb F}_{2^m}$ has the following orbit decomposition:
	\[
	\overline{\mathbb F}_{2^m}
	=
	O_0\cup O_1\cup \bigcup_{i=2}^{\frac{2^m+8}{12}}O_i,
	\]
	where
	$
	O_0=\{a\},
	O_1=
	\left\{
	\frac{b}{d},
	\frac{a+b}{1+d},
	\frac{a\omega+b}{\omega+d},
	\frac{a\omega^2+b}{\omega^2+d}
	\right\},
	$
	and
	$
	|O_i|=12,~2\le i\le \frac{2^m+8}{12}.
	$
 Moreover,
	$d\in \mathbb F_4$
	if and only if
	$\infty\in O_1.
	$
\end{proposition}
\begin{proof}
The proof is similar to Proposition \ref{orbit1}. So we omit it.
\end{proof}

In Case 2, a long orbit may contain $\infty$. If the selected
projective support does not contain $\infty$, then the construction yields
expurgated Goppa codes with \(A_4\) automorphism groups.
If the selected projective support contains $\infty$, then the construction
yields extended Goppa codes with \(A_4\) automorphism groups.

\begin{theorem}\label{2}
Let $\mathcal{L}=\bigcup_{i=2}^{s}O_i$, where $O_2,\ldots,O_s$,
$2\leq s\leq (2^m+8)/12$, are the orbits defined in
Proposition~\ref{4}. Let $\eta_1,\ldots,\eta_t$ be
elements in an extension field of $\mathbb{F}_{2^m}$, and let
$e_1,\ldots,e_t$ be positive integers. Assume that $G(\eta_j)$
does not contain $\infty$ and that
$G(\eta_j)\cap\mathcal{L}=\varnothing$ for $1\leq j\leq t$, and that
\[
g(X):=\prod_{j=1}^{t}
\left(\prod_{\alpha\in G(\eta_j)}(X-\alpha)\right)^{e_j}
\]
belongs to $\mathbb{F}_{2^m}[X]$. Then the following statements hold.

1) If $\infty\notin\mathcal{L}$, set $L=\mathcal{L}$. Then the expurgated Goppa code $\widetilde{\Gamma}(L,g)$ has an
automorphism group isomorphic to $A_4$.

2) If $\infty\in\mathcal{L}$, set
$L=\mathcal{L}\setminus\{\infty\}$. 
Then the extended Goppa code $\overline{\Gamma}(L,g)$ has an
automorphism group isomorphic to $A_4$.
\end{theorem}
\begin{proof}
The construction of the induced permutation group is the same as in
the proof of Theorem~\ref{5}. Let
\[
G=\langle A_P,B_P\rangle,\quad
A_P=PAP^{-1},\quad B_P=PBP^{-1}.
\]
By ordering the elements of each selected long orbit as in the proof
of Theorem~\ref{5}, the actions of $A_P$ and $B_P$ on the selected
projective coordinate set induce permutations $\psi_A$ and $\psi_B$
satisfying
$\langle\psi_A,\psi_B\rangle\cong A_4.$

It remains to verify that these coordinate permutations preserve the
corresponding codes. There are two differences from the affine case.
First, the elements $A_P$ and $B_P$ need not be affine
transformations. Let $M$ be either $A_P$ or $B_P$, and choose a
representative
\[
M=
\begin{pmatrix}
a_M&b_M\\
c_M&d_M
\end{pmatrix}.
\]
Since the root set of $g(X)$, with multiplicities, is stable under
$G$, there exists $\gamma_M\in\mathbb{F}_{2^m}^*$ such that
\[
(c_MX+d_M)^r
g\left(\frac{a_MX+b_M}{c_MX+d_M}\right)
=\gamma_M g(X),
\quad r=\deg g.
\]

Second, since $c_M$ may be nonzero, the transformed parity-check rows
may involve polynomials of degree $r$. We therefore use the
parity-check matrix of the expurgated Goppa code:
\[
\widetilde{H}=
\begin{pmatrix}
g(\alpha_1)^{-1}
    & g(\alpha_2)^{-1}
    & \cdots
    & g(\alpha_n)^{-1}\\
\alpha_1g(\alpha_1)^{-1}
    & \alpha_2g(\alpha_2)^{-1}
    & \cdots
    & \alpha_ng(\alpha_n)^{-1}\\
\vdots&\vdots&\ddots&\vdots\\
\alpha_1^r g(\alpha_1)^{-1}
    & \alpha_2^r g(\alpha_2)^{-1}
    & \cdots
    & \alpha_n^r g(\alpha_n)^{-1}
\end{pmatrix}.
\]
After applying the coordinate permutation induced by $M$, the column
indexed by $\alpha_i$ becomes
\[
\frac{1}{g(M(\alpha_i))}
\begin{pmatrix}
1\\
M(\alpha_i)\\
\vdots\\
M(\alpha_i)^r
\end{pmatrix}
=
\frac{\gamma_M^{-1}}{g(\alpha_i)}
\begin{pmatrix}
(c_M\alpha_i+d_M)^r\\
(a_M\alpha_i+b_M)(c_M\alpha_i+d_M)^{r-1}\\
\vdots\\
(a_M\alpha_i+b_M)^r
\end{pmatrix}.
\]
Every entry in the vector on the right is a polynomial in $\alpha_i$
of degree at most $r$. Hence every row of the transformed
parity-check matrix is a linear combination of the rows of
$\widetilde{H}$. Applying the same argument to $M^{-1}$ gives the
reverse inclusion, so the two matrices have the same row space.
Therefore, $\psi_A$ and $\psi_B$ preserve
$\widetilde{\Gamma}(L,g)$.

If $\infty\in\mathcal{L}$, we use the parity-check matrix
$\overline{H}$ of the extended Goppa code and regard its last column
as the column indexed by $\infty$. The same calculation in
projective coordinates, including the column corresponding to
$\infty$, shows that the transformed matrix has the same row space
as $\overline{H}$. Hence $\psi_A$ and $\psi_B$ also preserve
$\overline{\Gamma}(L,g)$.

Thus, in each case, the automorphisms induced by $\psi_A$ and
$\psi_B$ generate an automorphism group isomorphic to $A_4$.
\end{proof}

\begin{example}
We give examples in Tables~\ref{biao3} and~\ref{biao4}, using the same notation as in Case 1.
\begingroup
\scriptsize
\setlength{\tabcolsep}{2.2pt}
\renewcommand{\arraystretch}{1.12}

\begin{longtable}{C{0.8cm} C{1.9cm} C{2.7cm} P{9.0cm}}
\caption{Orbit decompositions in the non-affine case over \(\mathbb F_{2^6}\)}
\label{biao3}\\
\toprule
\multicolumn{1}{c}{No.}
& \multicolumn{1}{c}{Matrix \(P\)}
& \multicolumn{1}{c}{Short orbits}
& \multicolumn{1}{c}{Long orbits of length \(12\)}\\
\midrule
\endfirsthead

\multirow{5}{*}{1}
& \multirow{5}{*}{\(\begin{pmatrix}w^2&w^5\\1&w\end{pmatrix}\)}
& \multirow{5}{=}{\makecell[l]{fixed point: \(w^2\)\\short orbit:\\ \((w,w^4,w^{22},w^{60})\)}}
& \(L_1=(0,w^{59},w^{35},w^{12},w^{61},w^{25},w^{58},w^{51},w^{49},w^{38},w^{40},w^{21})\)\\
&&& \(L_2=(1,w^{41},w^{26},w^{45},w^{52},w^{14},w^7,w^{53},w^{13},w^{54},w^{15},w^3)\)\\
&&& \(L_3=(w^5,w^9,w^{47},w^{10},\infty,w^{34},w^{24},w^{55},w^{17},w^{46},w^{37},w^{36})\)\\
&&& \(L_4=(w^6,w^{43},w^{28},w^{16},w^{62},w^{56},w^8,w^{57},w^{44},w^{31},w^{48},w^{29})\)\\
&&& \(L_5=(w^{11},w^{20},w^{18},w^{19},w^{30},w^{27},w^{23},w^{33},w^{32},w^{50},w^{42},w^{39})\)\\
\midrule

\multirow{5}{*}{2}
& \multirow{5}{*}{\(\begin{pmatrix}w^2&w^5\\1&1\end{pmatrix}\)}
& \multirow{5}{=}{\makecell[l]{fixed point: \(w^2\)\\short orbit:\\ \((w,w^5,w^{17},\infty)\)}}
& \(L_1=(0,w^3,w^6,w^{14},w^{18},w^{23},w^{44},w^{49},w^{50},w^{53},w^{61},w^{62})\)\\
&&& \(L_2=(1,w^8,w^{11},w^{13},w^{21},w^{25},w^{28},w^{30},w^{31},w^{32},w^{51},w^{52})\)\\
&&& \(L_3=(w^4,w^9,w^{10},w^{22},w^{24},w^{34},w^{36},w^{37},w^{46},w^{47},w^{55},w^{60})\)\\
&&& \(L_4=(w^7,w^{12},w^{19},w^{20},w^{26},w^{29},w^{40},w^{42},w^{54},w^{56},w^{57},w^{59})\)\\
&&& \(L_5=(w^{15},w^{16},w^{27},w^{33},w^{35},w^{38},w^{39},w^{41},w^{43},w^{45},w^{48},w^{58})\)\\
\bottomrule
\end{longtable}
\endgroup

\begingroup
\scriptsize
\setlength{\tabcolsep}{2.2pt}
\renewcommand{\arraystretch}{1.16}

\begin{longtable}{C{1cm} P{2.2cm} C{6.2cm} C{1.55cm} C{1.35cm} P{3.3cm}}
\caption{Parameters of Goppa codes and related codes in the non-affine
case over $\mathbb{F}_{2^6}$}
\label{biao4}\\

\toprule
\multicolumn{1}{c}{No.}
& \multicolumn{1}{c}{Support}
& \multicolumn{1}{c}{\(g(X)\)}
& \multicolumn{1}{c}{\(\Gamma(L,g)\)}
& \multicolumn{1}{c}{Related code}
& \multicolumn{1}{c}{Remark}\\
\midrule
\endfirsthead

\caption[]{Parameters of Goppa codes and related codes in the non-affine
case over $\mathbb{F}_{2^6}$ (continued)}\\
\toprule
\multicolumn{1}{c}{No.}
& \multicolumn{1}{c}{Support}
& \multicolumn{1}{c}{\(g(X)\)}
& \multicolumn{1}{c}{\(\Gamma(L,g)\)}
& \multicolumn{1}{c}{Related code}
& \multicolumn{1}{c}{Remark}\\
\midrule
\endhead
\multirow{4}{*}{1}
& \(L_1\cup L_2\)
& \(X^4+w^{38}X^3+w^{40}X^2+w^{42}X+w^{24}\)
& \([24,4,12]\)
& \([24,4,12]\)
& \\

& \(L_1\cup L_2\cup L_4\)
& \(X^4+w^{38}X^3+w^{40}X^2+w^{42}X+w^{24}\)
& \([24,4,12]\)
& \([24,4,12]\)
& \\

& \(L_1\cup L_2\cup L_3\cup L_4\)
& \(X^4+w^{38}X^3+w^{40}X^2+w^{42}X+w^{24}\)
& \([47,23,9]\)
& \([48,23,10]\)
& extended Goppa code with an $A_4$ automorphism group\\

& \(L_1\cup\cdots\cup L_5\)
& \(X^4+w^{38}X^3+w^{40}X^2+w^{42}X+w^{24}\)
& \([59,35,9]\)
& \([60,35,10]\)
& extended Goppa code with an $A_4$ automorphism group\\
\midrule

\multirow{2}{*}{2}
& \multirow{3}{=}{\(L_1\cup\cdots\cup L_5\)}
& \makecell[l]{\(X^{12}+w^{27}X^9+w^{50}X^8+w^9X^6+w^{26}X^4\)\\
\(\quad +w^{54}X^3+w^{34}X^2+w^{14}X+w^{51}\)}
& \multirow{3}{*}{\([60,3,34]\)}
& \multirow{3}{*}{\([60,3,34]\)}
& \multirow{2}{=}{}
\\
\cmidrule(lr){3-3}
&
& \makecell[l]{\(X^{12}+w^{36}X^9+w^{13}X^8+w^{18}X^6+w^{42}X^4\)\\\(\quad+X^3
 +w^{43}X^2+w^{23}X+w^{15}\)}
&
&
& \\
\bottomrule
\\[-0.5ex]
\multicolumn{6}{p{0.95\textwidth}}{\emph{Note.} Each code in the column ``Related code'' has an $A_4$ automorphism
subgroup. The column ``Related code'' gives the parameters of
$\widetilde{\Gamma}(L,g)$, unless the Remark column indicates an extended
Goppa code, in which case it gives the parameters of
$\overline{\Gamma}(L,g)$. For a row in which the related code is
$\widetilde{\Gamma}(L,g)$, if $\Gamma(L,g)$ and
$\widetilde{\Gamma}(L,g)$ have the same dimension, then they coincide;
in this case, both codes have the stated $A_4$ automorphism subgroup.}
\end{longtable}

\endgroup    
\end{example}
 \subsection{Parameters of Goppa Codes with \(A_4\) Automorphism Groups}
We now study the parameters of a representative family of binary Goppa codes with an \(A_4\) automorphism subgroup. The Goppa polynomials considered here are powers of \(g(X)=X^4+X\). 


Before studying the parameters of our \(A_4\)-invariant Goppa codes, we recall a
dimension result of M. van der Vlugt for binary Goppa codes.
\begin{proposition}\label{prop:vdv-dimension}\cite{VanDerVlugt1990}
Let \(g(X)=g_1(X)^2g_2(X)\in \mathbb F_{2^m}[X]\), where \(g_2(X)\) is square-free.
Let \(r_i=\deg g_i(X)\) for \(i=1,2\), and let \(t\) be the number of distinct zeros of \(g(X)\) over the algebraic closure of \(\mathbb F_{2^m}\).
Let  $ Z=\{\alpha\in \mathbb F_{2^m}: g(\alpha)=0\}$,
    $L=\{\alpha\in \mathbb F_{2^m}: g(\alpha)\ne 0\}.$
Then
\[
    \dim_{\mathbb F_2}\Gamma(L,g)\ge 2^m-|Z|-m(r_1+r_2).
\]
Moreover, if
\[
    -2+\deg g+t<\frac{2^m+1-|Z|}{2^{m/2}},
\]
then
\[
    \dim_{\mathbb F_2}\Gamma(L,g)=2^m-|Z|-m(r_1+r_2).
\]
\end{proposition}
We first determine the exact parameters of a representative Goppa code arising from our \(A_4\)-invariant construction.
\begin{theorem}\label{main}
Let \(m\ge 6\) be even, let \(g(X)=X^4+X\), and let \(L=\{\alpha\in\mathbb F_{2^m}:g(\alpha)\ne0\}=\mathbb F_{2^m}\setminus\mathbb F_4\). Then the binary Goppa code \(\Gamma(L,g)\) has an automorphism subgroup isomorphic to \(A_4\) and has parameters \([2^m-4,\ 2^m-4m-4,\ 9]\).
\end{theorem}

\begin{proof}
 Since \(L=\mathbb F_{2^m}\setminus\mathbb F_4\), the length is \(n=|L|=2^m-4\). 
 
 We next explain why this code has an \(A_4\) automorphism subgroup. In the notation of Proposition~\ref{orbit1}, take \(a=1\) and \(b=0\). Then \(G=H_0\cong A_4\), and the orbit decomposition is
\[
\overline{\mathbb F}_{2^m}
=
O_0\cup O_1\cup \bigcup_{i=2}^{(2^m+8)/12}O_i,
\]
where \(O_0=\{\infty\}\), \(O_1=\mathbb F_4\), and the remaining orbits \(O_i\) have length \(12\). Therefore
\[
L=\mathbb F_{2^m}\setminus\mathbb F_4
=
\bigcup_{i=2}^{(2^m+8)/12}O_i
\]
is a union of long \(G\)-orbits.

Now apply Theorem~\ref{5} with one zero orbit, namely \(t=1\), \(\eta_1=0\), and \(e_1=1\). The \(G\)-orbit of \(\eta_1=0\) is precisely \(\mathbb F_4\). Hence
\[
\prod_{\alpha\in G(\eta_1)}(X-\alpha)
=
\prod_{\alpha\in\mathbb F_4}(X-\alpha)
=
X^4+X
=
g(X).
\]
Since this orbit is disjoint from \(L\), all the hypotheses of Theorem~\ref{5} are satisfied. Thus, the Goppa code \(\Gamma(L,g)\) has an automorphism subgroup isomorphic to \(A_4\).


It remains to determine the dimension and the minimum distance. Since \(g(X)\) is square-free, we apply Proposition~\ref{prop:vdv-dimension} with \(g_1(X)=1\), \(g_2(X)=g(X)\), \(r_1=0\), \(r_2=4\), \(Z=\mathbb F_4\), and \(t=4\). The equality condition becomes \(6<(2^m-3)/2^{m/2}\), which holds for every even \(m\ge6\). Thus \(\dim_{\mathbb F_2}\Gamma(L,g)=2^m-4-4m\). Moreover, this Goppa code belongs to the infinite family obtained from the fiber-type construction in Section~\ref{sec3}; hence Theorem~\ref{thm:x4cx-distance} directly gives \(d=9\).
\end{proof}
The next examples illustrate the parameter formula in Theorem~\ref{main} for small values of \(m\).
\begin{example}
The following examples were computed and verified by Magma, and they agree with Theorem~\ref{main}. For \(g(X)=X^4+X\) and \(L=\mathbb F_{2^m}\setminus\mathbb F_4\), we obtain
\[
m=6:\ [60,36,9],\qquad
m=8:\ [252,220,9],\qquad
m=10:\ [1020,980,9].
\]
\end{example}
Next, we consider powers of the same Goppa polynomial. The following lemma shows that, under a simple derivative condition, two consecutive powers define the same binary Goppa code.
\begin{lemma}\label{lem:g-r-equality}
Let \(L=\{\alpha_1,\ldots,\alpha_n\}\subseteq \mathbb F_{2^m}\), and let \(g(X)\in \mathbb F_{2^m}[X]\) satisfy \(g(\alpha_i)\ne0\) for \(1\le i\le n\). If \(g'(X)=\lambda\in\mathbb F_{2^m}^*\), then for every odd integer \(r\ge1\),
$\Gamma(L,g^r)=\Gamma(L,g^{r+1}).$
\end{lemma}
\begin{proof}
The inclusion \(\Gamma(L,g^{r+1})\subseteq\Gamma(L,g^r)\) is immediate. Conversely, let \(c=(c_1,\ldots,c_n)\in\Gamma(L,g^r)\), and put \(R_c(X)=\sum_{i=1}^n c_i/(X-\alpha_i)\). Then \(R_c(X)\equiv0\pmod {g^r}\), so we may write \(R_c(X)=g(X)^r u(X)\) in the localization in which all \(X-\alpha_i\) are invertible. Since the code is binary, \(R_c'(X)=\sum_{i=1}^n c_i/(X-\alpha_i)^2=R_c(X)^2\). Hence \(R_c'(X)\equiv0\pmod {g^{2r}}\), and in particular \(R_c'(X)\equiv0\pmod {g^{r+1}}\).

On the other hand, as \(r\) is odd and the characteristic is two, differentiating \(R_c(X)=g(X)^r u(X)\) gives \(R_c'(X)=g(X)^{r-1}(g'(X)u(X)+g(X)u'(X))=g(X)^{r-1}(\lambda u(X)+g(X)u'(X))\). Since \(R_c'(X)\equiv0\pmod {g^{r+1}}\), we have \(\lambda u(X)+g(X)u'(X)\equiv0\pmod {g^2}\). Reducing this congruence modulo \(g(X)\) gives \(u(X)\equiv0\pmod {g(X)}\), because \(\lambda\ne0\). Therefore \(R_c(X)\equiv0\pmod {g^{r+1}}\), and hence \(c\in\Gamma(L,g^{r+1})\). This proves the reverse inclusion.
\end{proof}
Based on the preceding lemma, we can also determine the dimension of another family of Goppa codes with an \(A_4\) automorphism subgroup arising from our construction.
\begin{theorem}\label{main2}
Let \(m\ge6\) be even, let \(g(X)=X^4+X\), and let \(L=\mathbb F_{2^m}\setminus\mathbb F_4\). Then \(\Gamma(L,g^3)=\Gamma(L,g^4)\). Moreover, this code has an automorphism subgroup isomorphic to \(A_4\), and its dimension \(k\) satisfies \(k\ge2^m-8m-4\). In particular, when \(m\ge8\), \(k=2^m-8m-4\).
\end{theorem}

\begin{proof}
Since \(g'(X)=1\), Lemma~\ref{lem:g-r-equality} with \(r=3\) gives \(\Gamma(L,g^3)=\Gamma(L,g^4)\). By the construction in Section~4, the support \(L=\mathbb F_{2^m}\setminus\mathbb F_4\) is a union of long orbits of the standard \(A_4\)-subgroup, and the zero set of \(g(X)\) is \(\mathbb F_4\). Hence, the code has an automorphism subgroup isomorphic to \(A_4\).

It remains to determine the dimension. Apply Proposition~\ref{prop:vdv-dimension} to \(g(X)^3=g_1(X)^2g_2(X)\), where \(g_1(X)=g_2(X)=g(X)\). Then \(r_1=r_2=4\), \(Z=\mathbb F_4\), and \(t=4\). Therefore \(k\ge2^m-4-8m\). The equality condition is \(14<(2^m-3)/2^{m/2}\), which holds for every even \(m\ge8\). Thus \(k=2^m-8m-4\) when \(m\ge8\).
\end{proof}
Then we provide some examples to illustrate Theorem \ref{main2}.
\begin{example}
Magma computations verify \(\Gamma(L,g^3)=\Gamma(L,g^4)\) and give
\[
m=6:\ [60,14,18],\qquad
m=8:\ [252,188,17],\qquad
m=10:\ [1020,940,17].
\]
These results are consistent with Theorem~\ref{main2}.
\end{example}
\subsection{Construction with \(A_5\) Automorphism Groups}

To obtain a matrix representation of \(A_5\), we first state the following conjugacy result for subgroups isomorphic to \(A_5\) in \(PGL(2,2^m)\).

\begin{lemma}\cite{Suzuki1982}\label{6}
Let \(PGL(2,2^m)\) contain a subgroup isomorphic to \(A_5\). Then all subgroups of \(PGL(2,2^m)\) isomorphic to \(A_5\) are conjugate to each other.
\end{lemma}

By this lemma, it remains to choose a standard representative of the conjugacy class of \(A_5\)-subgroups in \(PGL(2,2^m)\). Let \(\omega\in\mathbb F_4\) be a primitive element satisfying \(\omega^2+\omega+1=0\). Define
\[
A=
\begin{pmatrix}
\omega & 0\\
0 & \omega^2
\end{pmatrix},
\quad
B=
\begin{pmatrix}
1 & 1\\
0 & 1
\end{pmatrix},
\quad
C=
\begin{pmatrix}
0 & 1\\
1 & 0
\end{pmatrix}.
\]

\begin{theorem}
Define \(H_0=\langle A,B,C\rangle\le PGL(2,2^m)\). Then
$H_0=PGL(2,4)\cong A_5.$
Moreover,
\[
\{H\le PGL(2,2^m):H\cong A_5\}
=
\{PH_0P^{-1}:P\in PGL(2,2^m)\}.
\]
\end{theorem}
\begin{proof}
We first prove that \(H_0=PGL(2,4)\cong A_5\). Consider the natural action of \(H_0\) on \(\overline{\mathbb F}_4=\mathbb F_4\cup\{\infty\}\). Since the entries of \(A,B,C\) all lie in \(\mathbb F_4\), the set \(\overline{\mathbb F}_4\) is invariant under \(H_0\). The three generators act as \(A:z\mapsto \omega^2z\), \(B:z\mapsto z+1\), and \(C:z\mapsto 1/z\).

We first determine the elements of \(H_0\) fixing \(\infty\). Clearly \(A\) and \(B\) fix \(\infty\), so every element of \(\langle A,B\rangle\) fixes \(\infty\). Moreover, \(ABA^{-1}:z\mapsto z+\omega^2\) and \(A^2BA^{-2}:z\mapsto z+\omega\). Thus \(\langle A,B\rangle\) contains all translations \(z\mapsto z+\beta\), where \(\beta\in\mathbb F_4\). Together with the cyclic group generated by \(A\), we have
\[
\langle A,B\rangle
=
\{z\mapsto \alpha z+\beta:\alpha\in\mathbb F_4^*,\ \beta\in\mathbb F_4\}.
\]
Hence \(|\langle A,B\rangle|=|\mathbb F_4^*||\mathbb F_4|=12\). On the other hand, the elements of \(PGL(2,4)\) fixing \(\infty\) are precisely the affine transformations \(z\mapsto \alpha z+\beta\), where \(\alpha\in\mathbb F_4^*\) and \(\beta\in\mathbb F_4\). Since \(H_0\le PGL(2,4)\), it follows that the elements of \(H_0\) fixing \(\infty\) are exactly the elements of \(\langle A,B\rangle\).

Next, since \(C(\infty)=0\), the points \(0\) and \(\infty\) lie in the same \(H_0\)-orbit. Since \(\langle A,B\rangle\) contains all translations of \(\mathbb F_4\), it acts transitively on \(\mathbb F_4\). Therefore the \(H_0\)-orbit of \(\infty\) is \(\overline{\mathbb F}_4\), which has size \(5\). By the orbit-stabilizer theorem, \(|H_0|=5\cdot12=60\).

Since \(H_0\le PGL(2,4)\) and \(|PGL(2,4)|=4(4^2-1)=60\), we obtain \(H_0=PGL(2,4)\). The action of \(PGL(2,4)\) on \(\overline{\mathbb F}_4\) gives an injective homomorphism \(PGL(2,4)\to S_5\), because a nonidentity fractional linear transformation cannot fix three distinct points. Hence \(PGL(2,4)\) is isomorphic to a subgroup of \(S_5\) of order \(60\). Such a subgroup has index \(2\), and hence is normal in \(S_5\). Since the unique normal subgroup of \(S_5\) of order \(60\) is \(A_5\) \cite{GrouppropsS5}, we have \(PGL(2,4)\cong A_5\). Thus \(H_0=PGL(2,4)\cong A_5\).

By Lemma~\ref{6}, every subgroup \(H\le PGL(2,2^m)\) with \(H\cong A_5\) has the form \(H=PH_0P^{-1}\) for some \(P\in PGL(2,2^m)\). This completes the proof.
\end{proof}
We now use the above representation to construct Goppa codes and related codes with \(A_5\) automorphism groups.
Let \(H_0=\langle A,B,C\rangle\cong A_5\) be the standard subgroup of \(PGL(2,2^m)\) defined above. For
\[
P=
\begin{pmatrix}
a&b\\
c&d
\end{pmatrix}
\in PGL(2,2^m),
\]
set
\[
G=PH_0P^{-1}\le PGL(2,2^m).
\]
The orbit structure of \(G\) can be described uniformly for every choice of \(P\).

\begin{theorem}\label{thm:A5-orbits}
The action of \(G=PH_0P^{-1}\) on \(\overline{\mathbb F}_{2^m}=\mathbb F_{2^m}\cup\{\infty\}\) has the following orbit decomposition.

Put
\[
O_0=
\left\{
\frac{a}{c},
\frac{b}{d},
\frac{a+b}{c+d},
\frac{a\omega+b}{c\omega+d},
\frac{a\omega^2+b}{c\omega^2+d}
\right\},
\]
where a quotient with zero denominator is interpreted as \(\infty\).

If \(4\nmid m\), then
\[
\overline{\mathbb F}_{2^m}
=
O_0\cup \bigcup_{i=1}^{(2^m-4)/60}O_i,
\]
where \(|O_i|=60\) for \(i\ge1\).

If \(4\mid m\), then
\[
\overline{\mathbb F}_{2^m}
=
O_0\cup O_1\cup \bigcup_{i=2}^{1+(2^m-16)/60}O_i,
\]
where \(|O_1|=12\) and \(|O_i|=60\) for \(i>1\).

Moreover, \(\infty\in O_0\) if and only if either \(c=0\) or \(c\ne0\) and \(d/c\in\mathbb F_4\). If \(4\mid m\), then \(\infty\in O_1\) if and only if \(c\ne0\) and \(d/c\in\mathbb F_{16}\setminus\mathbb F_4\). 
\end{theorem}
\begin{proof}
 We first consider the action of \(H_0\) on $\overline{\mathbb F}_{2^m}$.  Since the entries of \(A,B,C\) lie in \(\mathbb F_4\), the set \(\overline{\mathbb F}_4=\mathbb F_4\cup\{\infty\}\) is invariant under \(H_0\). The generators act as \(A:z\mapsto \omega^2z\), \(B:z\mapsto z+1\), and \(C:z\mapsto 1/z\). Since \(C(\infty)=0\), and since \(\langle A,B\rangle\) contains all affine transformations \(z\mapsto \alpha z+\beta\), where \(\alpha\in\mathbb F_4^*\) and \(\beta\in\mathbb F_4\), we have \(H_0(\infty)=\overline{\mathbb F}_4\). Thus, \(\overline{\mathbb F}_4\) is an orbit of length \(5\).

We next analyze the stabilizers of the points outside this orbit. Since \(H_0\cong A_5\), every nonidentity element of \(H_0\) has order \(2\), \(3\), or \(5\). The element \(B:z\mapsto z+1\) fixes only \(\infty\). Since all involutions in \(A_5\) are conjugate, every element of order \(2\) in \(H_0\) fixes only a point of \(\overline{\mathbb F}_4\). Similarly, \(A:z\mapsto \omega^2z\) fixes exactly \(0\) and \(\infty\). Since all elements of order \(3\) in \(A_5\) are conjugate, every element of order \(3\) in \(H_0\) has all its fixed points in \(\overline{\mathbb F}_4\).

It remains to consider elements of order \(5\). Let \(T_\omega:z\mapsto z+\omega\), and set \(D=CT_\omega\). Then \(D\in H_0\) and \(D:z\mapsto 1/(z+\omega)\). On \(\overline{\mathbb F}_4\), this element acts as the cycle \((\infty,0,\omega^2,1,\omega)\), and hence \(D\) has order \(5\). Its fixed points satisfy \(z=1/(z+\omega)\), or equivalently \(z^2+\omega z+1=0\). This polynomial has no root in \(\mathbb F_4\), and its two roots lie in \(\mathbb F_{16}\setminus\mathbb F_4\). All nonidentity elements of a cyclic subgroup of order \(5\) have the same two fixed points. Since all Sylow \(5\)-subgroups of \(H_0\) are conjugate, and since \(H_0\) preserves both \(\overline{\mathbb F}_{16}\) and \(\overline{\mathbb F}_4\), the fixed points of any element of order \(5\) in \(H_0\) lie in \(\overline{\mathbb F}_{16}\setminus\overline{\mathbb F}_4\).

We now distinguish two cases.

\begin{enumerate}
\item If \(4\nmid m\), then \(\mathbb F_{2^m}\cap\mathbb F_{16}=\mathbb F_4\). Hence no point of \(\overline{\mathbb F}_{2^m}\setminus\overline{\mathbb F}_4\) can be fixed by an element of order \(5\). By the preceding discussion, such a point also cannot be fixed by an element of order \(2\) or \(3\). Therefore, its stabilizer in \(H_0\) is trivial, and every remaining orbit has length \(|H_0|=60\). Since \(|\overline{\mathbb F}_{2^m}|=2^m+1\), after removing the orbit \(\overline{\mathbb F}_4\) there remain \(2^m-4\) points. Thus, the number of remaining orbits is \((2^m-4)/60\).

\item If \(4\mid m\), then \(\mathbb F_{16}\subseteq\mathbb F_{2^m}\). Put \(O_1'=\overline{\mathbb F}_{16}\setminus\overline{\mathbb F}_4\). Choose a root \(\alpha_0\in\mathbb F_{16}\setminus\mathbb F_4\) of \(z^2+\omega z+1\). Then \(D(\alpha_0)=\alpha_0\), so \(\langle D\rangle\le (H_0)_{\alpha_0}\). Moreover, no element of order \(2\) or \(3\) fixes \(\alpha_0\). Since \(|(H_0)_{\alpha_0}|\) divides \(60\), is divisible by \(5\), and is divisible by neither \(2\) nor \(3\), we have \(|(H_0)_{\alpha_0}|=5\). The orbit-stabilizer theorem then gives \(|H_0(\alpha_0)|=60/5=12\). Since \(|O_1'|=|\overline{\mathbb F}_{16}|-|\overline{\mathbb F}_4|=17-5=12\), we have \(H_0(\alpha_0)=O_1'\). For any point outside \(\overline{\mathbb F}_4\cup O_1'\), the stabilizer is trivial, and hence every remaining orbit has length \(60\). The number of such orbits is \((2^m+1-5-12)/60=(2^m-16)/60\).
\end{enumerate}

Finally, the orbit decomposition for \(G=PH_0P^{-1}\) is obtained by applying \(P\) to the orbit decomposition of \(H_0\). In particular,
\[
P(\overline{\mathbb F}_4)
=
\left\{
\frac{a}{c},
\frac{b}{d},
\frac{a+b}{c+d},
\frac{a\omega+b}{c\omega+d},
\frac{a\omega^2+b}{c\omega^2+d}
\right\},
\]
which gives the stated orbit \(O_0\). When \(4\mid m\), the unique orbit of length \(12\) is carried to the orbit \(O_1\), while all remaining orbit lengths are preserved. Moreover,
\[
P^{-1}(\infty)
=
\begin{cases}
\infty,&c=0,\\
d/c,&c\ne0.
\end{cases}
\]
The stated conditions for the orbit containing \(\infty\) now follow from the orbit decomposition of \(H_0\).
\end{proof}
Using unions of orbits in this decomposition as projective supports, we obtain expurgated and extended Goppa codes with \(A_5\) automorphism groups, according as the selected projective support does not or does contain \(\infty\).
\begin{theorem}\label{thm:A5-construction}
Let \(\mathcal L\) be a nonempty union of orbits of length \(60\)
in the above decomposition. Let $\eta_1,\ldots,\eta_t$ be elements in an extension field of $\mathbb{F}_{2^m}$, and let $e_1,\ldots,e_t$ be positive integers. Assume that \(G(\eta_j)\) does not contain \(\infty\) and that \(G(\eta_j)\cap\mathcal{L}=\varnothing\) for \(1\leq j\leq t\), and that
\[
g(X):=\prod_{j=1}^t
\left(\prod_{\alpha\in G(\eta_j)}(X-\alpha)\right)^{e_j}
\]
belongs to \(\mathbb{F}_{2^m}[X]\).
Then the following statements hold.
\begin{enumerate}
\item If $\infty\notin\mathcal{L}$, set $L=\mathcal{L}$. Then the
      expurgated Goppa code $\widetilde{\Gamma}(L,g)$ has an automorphism group isomorphic to $A_5$.
\item If $\infty\in\mathcal{L}$, set
      $L=\mathcal{L}\setminus\{\infty\}$. Then the extended Goppa code $\overline{\Gamma}(L,g)$ has an automorphism group isomorphic to $A_5$.
\end{enumerate}
\end{theorem}
\begin{proof}
The proof is analogous to that of Theorem~\ref{2} and is omitted.
\end{proof}
Since $A_4\leq A_5$, an expurgated or extended Goppa code with an $A_5$
automorphism group can be regarded as a special case of the
$A_4$-construction. The following proposition explains the relation between
long $A_4$-orbits and long $A_5$-orbits.


\begin{proposition}\label{8}
Let \(G\cong A_5\) act on a set \(\Omega\), and let \(H\le G\) satisfy \(H\cong A_4\) and \([G:H]=5\). Choose an element \(R\in G\) of order \(5\). If \(\xi\in\Omega\) satisfies \(|G(\xi)|=60\), then
\[
G(\xi)
=
H(\xi)\sqcup RH(\xi)\sqcup R^2H(\xi)\sqcup R^3H(\xi)\sqcup R^4H(\xi).
\]  
\end{proposition}
\begin{proof}
Since \(H\cong A_4\), the subgroup \(H\) contains no element of order \(5\). Hence the left cosets \(R^iH~(i=0,1,2,3,4)\) are pairwise distinct. Since \([G:H]=5\), they give the full left-coset decomposition
\[
    G=H\sqcup RH\sqcup R^2H\sqcup R^3H\sqcup R^4H .
\]
It follows that
\[
    G(\xi)=H(\xi)\cup RH(\xi)\cup R^2H(\xi)\cup R^3H(\xi)\cup R^4H(\xi).
\]

It remains to show that this union is disjoint. Since \(|G(\xi)|=60=|G|\), the orbit-stabilizer theorem gives \(G_\xi=\{1\}\). In particular, \(H_\xi=\{1\}\), and hence \(|H(\xi)|=|H|=12\). Moreover, for each \(0\le i\le4\), the set \(R^iH(\xi)\) also has size \(12\).

Suppose that \(R^ih_1(\xi)=R^jh_2(\xi)\) for some \(h_1,h_2\in H\). Since \(G_\xi=\{1\}\), we have \(R^ih_1=R^jh_2\). Thus, the two left cosets \(R^iH\) and \(R^jH\) intersect, and hence they are equal. By the pairwise distinctness of the five left cosets, this implies \(i=j\). Therefore, the five sets above are pairwise disjoint. The proposition follows.
\end{proof}
\begin{example}
Take \(P=I_2\), so that \(G=H_0\), and let
\(L=\mathbb{F}_{2^8}\setminus\mathbb{F}_{16}\) and
\(g(X)=X^{12}+X^9+X^6+X^3+1\). Since
\((X^4+X)g(X)=X^{16}+X\), the zero set of \(g(X)\) is
\(\mathbb{F}_{16}\setminus\mathbb{F}_4\).

For \(\eta\in\mathbb{F}_{16}\setminus\mathbb{F}_4\), Theorem~\ref{thm:A5-orbits}
gives \(H_0(\eta)=\mathbb{F}_{16}\setminus\mathbb{F}_4\), which has
length \(12\) and does not contain \(\infty\). Moreover,
\(L=O_2\cup O_3\cup O_4\cup O_5\), where \(O_2,\ldots,O_5\) are the
\(H_0\)-orbits of length \(60\). Thus,
\(H_0(\eta)\cap L=\varnothing\), and
\(g(X)=\prod_{\alpha\in H_0(\eta)}(X-\alpha)\) belongs to
\(\mathbb{F}_{2^8}[X]\). Therefore, Theorem~\ref{thm:A5-construction} implies that
\(\widetilde{\Gamma}(L,g)\) has an automorphism group isomorphic to
\(A_5\).

Let \(H=\langle A,B\rangle\cong A_4\). By Proposition~\ref{8}, each
\(O_i\), \(2\leq i\leq 5\), is a disjoint union of five long
\(H\)-orbits. Consequently, \(L\) is also a disjoint union of twenty
long \(A_4\)-orbits.
\end{example}
\section{Conclusion}~\label{sec5}
In this paper, we studied the minimum distance of binary separable Goppa codes and the construction of binary Goppa codes with prescribed automorphism groups. We established two criteria for Goppa codes defined by power-composite polynomials \(g(X)=f(X^t)\) and composite polynomials \(g(X)=A(X)h(\phi(X))\) to attain their designed distances, and derived several infinite families with determined minimum distance. We also constructed binary Goppa codes, expurgated Goppa codes, and extended Goppa codes with \(A_4\) or \(A_5\) automorphism groups. These constructions naturally yield binary quasi-cyclic Goppa codes and related codes, and for a representative \(A_4\)-invariant family, we determined the parameters by applying the minimum-distance criteria developed above.

Several problems remain open. For example, for the second minimum-distance criterion, concerning Goppa polynomials of the form
\(g(X)=A(X)h(\phi(X))\), it would be interesting to find further simple explicit families beyond \(g(X)=X^4+cX\). We will investigate this problem further in our future work.

\section*{Acknowledgment}
Part of this work has been submitted to the conference Sequences and Their Applications (SETA) 2026.
    \bibliographystyle{plain}
    \bibliography{ref.bib}


 \end{document}